\documentclass[a4paper, 3p]{elsarticle} 
\usepackage{amsmath,amssymb, amsthm,  graphicx} 
\usepackage{lineno}
\usepackage{setspace}
\usepackage{xcolor} 
\usepackage{float}  

\newtheorem{thm}{Theorem}[section]

\theoremstyle{definition}
\newtheorem{defn}[thm]{Definition}
\newtheorem{rem}[thm]{Remark}
\newenvironment{pf}{\begin{proof}}{\end{proof}}

\begin{document}

\begin{frontmatter}
\title{A Neural Network Framework for Geodesic-Like Curve Computation on Parametric Surfaces}
\journal{Computer Aided Geometric Design} 

\author[NCYU]{Sheng-Gwo Chen\corref{cor1}}\ead{csg@mail.ncyu.edu.tw}
\author[NCYU]{Chen-Chang Peng}\ead{ccpeng@mail.ncyu.edu.tw}

\cortext[cor1]{Corresponding author}

\address[NCYU]{Department of Applied Mathematics, National Chiayi University, Chia-Yi 600, Taiwan}

\begin{abstract}
    The concept of geodesic-like curves was introduced by Chen in 2010 as a method for estimating shortest
    paths (geodesics) on parametric surfaces, with its convergence established theoretically.
    However, an efficient numerical computational framework has not yet been developed.
    In this paper, we propose an elegant and efficient approach for
    computing geodesic-like curves by leveraging deep learning and
    Physics-Informed Neural Networks (PINNs). Under the proposed
    framework, not only can single parametric surfaces be handled
    efficiently, but a broad class of complex parametric surfaces
    including multi-surface systems with $C^0$ or higher continuity and surfaces of revolution can also be robustly addressed.
\end{abstract}

\end{frontmatter}
 \section{Introduction}

The computation of geodesics on smooth manifolds is a fundamental
problem in differential geometry, with important applications in
Computer-Aided Geometric Design, robotics, and computer vision,
among other fields.\cite{Paluszny,Ravi2,Sanchez,Sprynski}
A geodesic
is mathematically defined as the shortest path between two points on
a surface.\cite{Carmo,Carmo2} Such curves play a crucial role in
trajectory optimization for autonomous navigation and in modeling
wave propagation in complex physical environments.
Despite their
broad applicability, the efficient computation of geodesics on
arbitrary parametric surfaces remains a challenging problem.

Regarding numerical approaches for geodesic computation on surfaces,
Beck et al.\cite{Beck1} computed geodesics on spline surfaces using
a fourth-order Runge-Kutta method, while Sneyd and
Reskin\cite{Sneyd1} improved this approach through a second-order
Runge-Kutta scheme. Hotz and Hange\cite{Hotz}, as well as Peng et
al.\cite{Zhang1}, proposed geometric approaches for computing
geodesics on arbitrary surfaces. Ying and Candes\cite{Ying1} introduced a
phase-flow method for computing geodesic flows on smooth and compact
surfaces. Since the shortest smooth path between two points on a
regular surface is a geodesic, efficient numerical approximations
are of considerable interest. In 2005, Kasap et al.\cite{Kasap}
presented a finite difference method for computing geodesics between
two points on a surface. Subsequently, in 2010, Chen\cite{Chen}
proposed the concept of geodesic-like curves to approximate shortest
smooth paths on parametric surfaces, and their convergence was
established theoretically. In addition, numerous accurate discrete
algorithms for geodesic computation on
surfaces,\cite{Bulut1,Meng1,Ravi3,Tucker} polygonal
surfaces,\cite{Kanai1,Polt} and triangular
meshes\cite{Bose1,Chen2,Crane1,Li1,Martinez,Surazhsky} have been
developed.

Recently, Scientific Machine Learning (SciML) has emerged as a
promising paradigm for solving nonlinear geometric problems. Neural
networks, as powerful universal function approximators, enable the
solution of partial differential equations without relying on
traditional grid-based discretization. In particular,
Physics-Informed Neural Networks (PINNs)\cite{Raissi1,Waheed1} and
the Deep Ritz method\cite{Rowan1} have been employed to minimize
energy functionals for geodesic computation. Furthermore,
Zhang\cite{Zhang2} proposed the Neural Geodesic Field (NeuroGF), a
neural-network-based framework for computing geodesics on
three-dimensional mesh models.

In this paper, we propose a neural-network-based framework for
geodesic estimation on parametric surfaces using geodesic-like
curves.\cite{Chen} The theory of geodesic-like curves, introduced by
Chen, provides a sequence of approximating curves that converges to
the true geodesic as the curve order increases. Instead of directly
solving the geodesic differential equations, the proposed framework
employs an artificial neural network to optimize the parameters of
geodesic-like curves. By combining the theoretical convergence
properties of geodesic-like curves with the optimization capability
of neural networks, the proposed approach enables accurate and
efficient estimation of two-point geodesics on parametric surfaces.

  \section{Geodesics and Geodesic-like Curves on Regular Surfaces}

  \subsection*{Geodesics on Regular Surfaces}
  Let $\Sigma$ be a regular surface with a parametrization $\mathbf{x} : U \subset \mathbb{R}^2 \rightarrow \Sigma$
  where $U$ is an open subset of $\mathbb{R}^2$.
  Let $\gamma:[a,b]\to U$ be a smooth curve on $\Sigma$.
  Then the arc length of $\gamma$ and the energy of $\gamma$ are given by
  \begin{equation}
    L(\gamma) = \int_a^b \| (\mathbf{x} \circ \gamma)'(t) \| dt
  \end{equation}
  and
  \begin{equation}
    E(\gamma) = \frac{1}{2} \int_a^b \| (\mathbf{x} \circ \gamma)'(t) \|^2 dt,
  \end{equation}
  respectively. A proper variation of $\gamma$ is a
  differentiable map $\phi :[-\epsilon, \epsilon] \times [a,b]
  \rightarrow U$ such that $\phi_s(\cdot) := \phi(s, \cdot)$ is a
  smooth curve from $\gamma(a)$ to $\gamma(b)$ for each
  $s \in [-\epsilon, \epsilon]$ and $\phi_0 = \gamma$; see
  Figure \ref{fig_variation}. Specifically, it satisfies
  \begin{equation}
  \begin{array}{ll}
    \phi_0(t) = \phi(0,t)= \gamma(t), & t \in [a,b], \\
    \phi_s(a) = \phi(s,a)=\gamma(a), & s \in [-\epsilon, \epsilon], \\
    \phi_s(b) = \phi(s,b)=\gamma(b), & s \in [-\epsilon, \epsilon].
  \end{array}
  \end{equation}

  The arc length and energy of the variation $\phi_s$ are given by
  \begin{equation}
    L(s) = \int_a^b \| (\mathbf{x} \circ \phi_s)'(t) \| dt,
  \end{equation}
  and
  \begin{equation}\label{energy_functional}
    E(s) = \frac{1}{2} \int_a^b \| (\mathbf{x} \circ \phi_s)'(t) \|^2 dt
  \end{equation}
  respectively.
  Let $V(t) = \frac{\partial }{\partial s} ( \mathbf{x} \circ \phi)(s,t) |_{s=0}$ denote the variational vector field
  of $ \mathbf{x} \circ \phi $. The first variation formula of $E(s)$
  is given by
  \begin{equation}
  \begin{aligned} E'(0)  & = \int _a^b \left \langle \frac{\partial }{\partial t} (\mathbf{x} \circ \phi) (0,t),
                          \frac{D}{dt}V(t) \right \rangle dt \cr\cr
                      & = - \int_a^b \left \langle   \frac{D}{dt} ( \frac{d}{dt}( \mathbf{x} \circ \gamma)  ) , V(t) \right
                      \rangle  dt.
    \end{aligned}
  \end{equation}
  If the curve $\gamma$ satisfies the equation
  $$\frac{D}{dt} \frac{d}{dt} ( \mathbf{x} \circ \gamma) = 0\mbox{ for all }t \in [a,b],$$
  then  $\mathbf{x} \circ \gamma$ is called a geodesic on $\Sigma$.
  For simplicity, $\gamma$ is also referred to as a geodesic.
  The following theorem is immediate\cite{Carmo,Carmo2}:

  \begin{figure}
  \centering
  \includegraphics[width=0.5\textwidth]{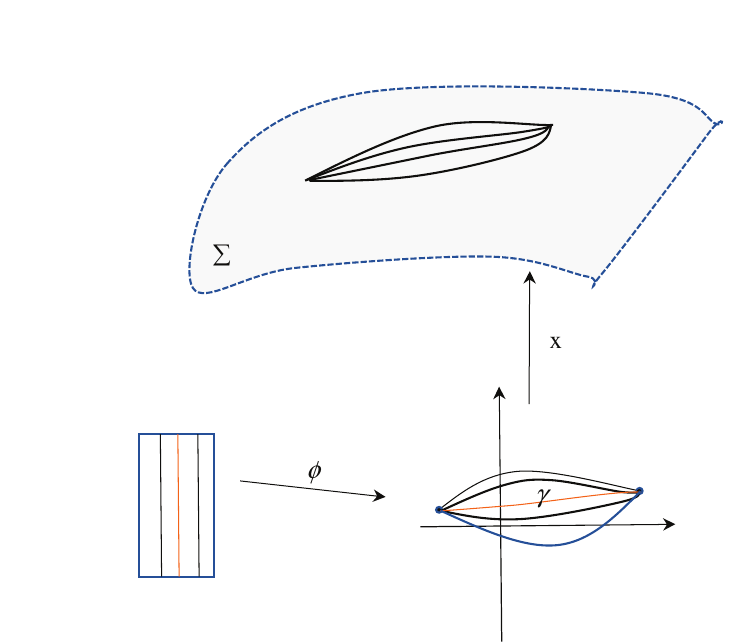}
  \caption{The proper variation of a curve $\gamma$ on $\Sigma$.}\label{fig_variation}
  \end{figure}

  \begin{thm}
  Let $\Sigma$ be a regular surface with a parametrization $\mathbf{x} : U \subset \mathbb{R}^2 \rightarrow \Sigma$ and $\gamma$
  a smooth curve on $\Sigma$. The following statements are equivalent:
  \begin{enumerate}
    \item $\mathbf{x} \circ \gamma$ is a geodesic on $\Sigma$.
    \item $\gamma$ is a critical point of the energy functional, that is, $E'(0)=0$.
    \item $\gamma$ is a critical point of the arc length functional, that is, $L'(0)=0$.
  \end{enumerate}

  \end{thm}

  The shortest path between two distinct points $\mathbf{p},  \mathbf{q} \in \Sigma$
   is defined as the solution of the following minimization problem:
  \begin{equation}
    \min \{ L(\gamma) : \gamma \mbox{ is a smooth curve on } \Sigma \mbox{ with } \gamma(a) = \mathbf{p}, \gamma(b)
    = \mathbf{q} \}.
  \end{equation}
  It is also a critical point of the energy functional defined in Equation (\ref{energy_functional}),
  and hence a geodesic on $\Sigma$.

  \subsection*{Geodesic-like Curves on Regular Surfaces}

  The shortest path between two distinct points on $\Sigma$
  is sought within the space $C([a,b], U)$.
  In practice, one typically restricts attention to curves
   with favorable geometric properties,
   such as B\'ezier curves, B-spline curves, NURBS curves,
   and Fourier polynomial representations,
   rather than considering all smooth (or continuous) curves.

  Let $\mathcal{V}$ be a subset of $C([a,b],U)$.
  The minimal geodesic-like curve on $\Sigma$ with respect to
  $\mathcal{V}$ is defined as follows:

  \begin{defn}
  Let $\Sigma$ be a regular surface with a parametrization $\mathbf{x} : U \subset \mathbb{R}^2 \rightarrow \Sigma$ and
  $\mathcal{V}$ be a family of smooth curves from $[a,b]$ into $U$, i.e., $\mathcal{V}\subset C^{\infty}([a,b],U)$.
  Suppose that $\mathbf{p}, \mathbf{q}\in \mathbf{x}(U)$ are two
  distinct points on $\Sigma$.  A minimal geodesic-like curve between $\mathbf{p}$ and $\mathbf{q}$ on $\Sigma$
  with respect to $\mathcal{V}$
  is a curve $\gamma \in \mathcal{V}$ satisfying
  $\mathbf{x}(\gamma(a)) = \mathbf{p}$, $\mathbf{x}(\gamma(b)) = \mathbf{q}$
  such that
  \begin{equation}
    E(\gamma) = \min \{ E(\alpha) : \alpha \in \mathcal{V},\mathbf{x}(\alpha(a)) = \mathbf{p}, \mathbf{x}(\alpha(b)) = \mathbf{q} \}.
  \end{equation}
  \end{defn}

  From a topological perspective, \cite{Munkres} the following theorem is immediate:

  \begin{thm}\label{thm_geodesic_like_curve1}
  Let $\Sigma$ be a regular surface with a parametrization
  $\mathbf{x} : U \subset \mathbb{R}^2 \rightarrow \Sigma$, and let
  $\gamma : [a,b] \rightarrow U$ be a geodesic on $\Sigma$
  without conjugate points.
  Suppose that $\{\mathcal{V}_n\}_{n=1}^{\infty}$
  is a sequence of subsets of $C([a,b],U)$ satisfying
  $$\mathcal{V}_n \subset \mathcal{V}_{n+1},~n \in \mathbb{N}$$
  and that
  $\mathcal{V} = \bigcup\limits_{n=1}^{\infty} \mathcal{V}_n$
  is dense in $C([a,b], U)$ with respect to the uniform topology.
  Then there exists a sequence of minimal geodesic-like curves
  $\{\gamma_n  \in \mathcal{V}_n\}_{n=1}^{\infty}$ converges to
  $\gamma$ as $n\to\infty$.
  \end{thm}
  \begin{pf}
  Suppose that $\gamma$ is a geodesic on $\Sigma$ satisfying
  $\gamma(a) = \mathbf{p}$ and $\gamma(b) = \mathbf{q}$.
  For each positive integer $n$, let $\gamma_n$ denote a minimal
  geodesic-like curve between $\mathbf{p}$ and $\mathbf{q}$
  on $\Sigma$ with respect to $\mathcal{V}_n$. We show that
  $$\lim\limits_{n \rightarrow \infty} \gamma_n = \gamma.$$
  Since $\mathcal{V}$ is dense in $C([a,b], U)$,
  there exists a sequence of curves
  $\{ \alpha_n \in \mathcal{V}_n\}_{n=1}^{\infty}$
  converging uniformly to $\gamma$.
  By the definition of the minimal geodesic-like curve, we have
  \begin{equation}
  E(\gamma_n) \leq E(\alpha_n) \text { for each }  n \in \mathbb{N}.
  \end{equation}
  Since $\alpha_n\to\gamma$ uniformly,
  continuity of the energy functional implies
  \begin{equation}
     \lim_{n \rightarrow \infty} E(\alpha_n)
     = E(\gamma).
     \end{equation}
  Furthermore, because $\gamma$ is a geodesic without conjugate points,
   it is a local minimizer of the energy functional. Therefore,
  \begin{equation}
    E(\gamma) \leq \liminf_{n\to\infty} E(\gamma_n)
    \leq \limsup_{n\to\infty} E(\gamma_n)
      \leq \lim_{n \rightarrow \infty} E(\alpha_n) = E(\gamma).
  \end{equation}
  Hence,
  $$
  \lim_{n \rightarrow \infty}E\left(\gamma_n\right)
  = E\left(\gamma\right).
  $$
  By the uniqueness of the local minimizer in the absence of
  conjugate points, it follows that
  \begin{equation}
  \gamma_n\to\gamma\mbox{ as }n\to\infty.
  \end{equation}
  \end{pf}

  If $\mathcal{V}_n$ in theorem \ref{thm_geodesic_like_curve1}
  denote the set of all B\'ezier (resp. B-spline, NURBS) curves of
  degree $n$ on $U$, then the corresponding minimal geodesic-like
  curve is referred to as a minimal B\'ezier (resp. B-spline, NURBS)
  geodesic-like curve.

  \begin{rem}
  \begin{enumerate}
  \item  By the Weierstrass approximation theorem\cite{Marsden}, the set of all B\'ezier curves on $U$ is dense in $C([a,b], U)$. Hence,
      any geodesic on $\mathbf{x}(U)$ can be approximated by a sequence of minimal B\'ezier geodesic-like curves.

  \item Since every B\'ezier curve can be regarded as a special case of a B-spline (respectively, NURBS) curve,
    any geodesic $\mathbf{x}(U)$ can also be approximated by a sequence of minimal B-spline (respectively, NURBS) geodesic-like curves.

    \item Fixed a positive integer $n$, the set of all uniform B-spline curves of degree $n$ on $U$ is also dense in $C([a,b],
    U)$.
    Hence any geodesic on $\mathbf{x}(U)$ can be approximated by a  sequence of minimal uniform B-spline geodesic-like
    curves of degree $n$.
    In fact, a discrete set of points representing a numerical shortest path may be interpreted
    as a degree-1 minimal B-spline geodesic-like curve, where the points serve as control points.

  \end{enumerate}

  \end{rem}

  For further details on geodesic-like curves, readers are referred to Chen.\cite{Chen}
  In this paper, we focus on the uniform B-spline geodesic-like curves on a regular surface
  $\Sigma$,
  using them to approximate the shortest path between two points on $\Sigma$.
   Let $[a,b] \subset \mathbb{R}$ be a parametric interval, and let $\{\mathbf{b}_i\}_{i=0}^{n} \subset U \subset \mathbb{R}^2$
  be a collection of $(n+1)$ control points. Given a clamped uniform knot vector $K=\{t_0,\dots,t_{n+d+1}\}$ on $[a,b]$,
  a clamped uniform B-spline curve $\gamma:[a,b]\to U$ of degree $d$ is defined by \[
  \gamma(t)= \sum_{i=0}^{n} N_{i,d}(t)\mathbf{b}_i, \]
  where $\{N_{i,d}(t)\}_{i=0}^{n}$ denotes the set of B-spline basis functions
  associated with the knot vector $K$.

  To ensure that the curve interpolates the endpoints $\mathbf{b}_0$
and $\mathbf{b}_n$, the knot vector is constructed such that the
first and last $d+1$ knots coincide at $a$ and $b$, respectively:
\begin{equation}\label{eqn_knot_vector}
K = \{\underbrace{a,\dots,a}_{d+1}, t_{d+1},\dots,t_{n},
\underbrace{b,\dots,b}_{d+1}\}.
\end{equation}

The interior knots $\{t_j\}_{j=d+1}^{n}$ are uniformly spaced with
step size $\Delta t = \frac{b-a}{n-d}$.

  \begin{defn}
  Let $\mathcal{V}$ be the set of all B-spline curves of degree $d$ on the interval $[a,b]$ with control points
  in $U \subseteq \mathbb{R}^m$.
  A curve $\gamma \in \mathcal{V}$ is represented as
  \begin{equation}
  \gamma(t) = \sum_{i=0}^{n} \mathbf{b}_i \, N_{i,d}(t), \quad t \in [a,b],
  \end{equation}
  where $\mathbf{b}_i \in U$ are the control points and $N_{i,d}(t)$ denote the B-spline basis functions of degree $d$
  defined on a given knot vector.

  A \emph{$\mathcal{V}$-admissible variation of $\gamma$} is a smooth map
  \begin{equation}
  \phi : [-\epsilon,\epsilon] \times [a,b]\longrightarrow U
  \end{equation}
  such that for each $s \in [-\epsilon,\epsilon]$, the curve $\phi_s(\cdot) := \phi(s,\cdot)$ belongs to $\mathcal{V}$ and
  $\phi_0(t) = \gamma(t)$ for all $t \in [a,b]$. Equivalently,
  \begin{equation}
  \phi(s,t) = \sum_{i=0}^{n} \mathbf{b}_i(s)\, N_{i,d}(t),
  \end{equation}
  where the control points $\mathbf{b}_i(s)$ depend smoothly on the variation parameter $s$.
  \end{defn}

    The energy functional of $\phi(s,t)$ on the surface $\Sigma$ is defined as
\begin{equation}\label{eqn_energy_variation2}
 E(s)= \frac{1}{2} \int_a^b \left\| \frac{\partial}{\partial t}
(\mathbf{x} \circ \phi)(s,t) \right\|^2 dt.
\end{equation}

  Since B-spline curves are completely determined by their control points,
  the energy of the $\mathcal{V}$-admissible variation $\phi(s,t)$ in Equation (\ref{eqn_energy_variation2}) can be expressed
  in terms of the control points $\mathbf{b}_i(s)$. It can be written as
  \begin{equation}
    E(u_0,u_1, \dots, u_n,v_0,v_1, \dots, v_n) =
    \frac{1}{2} \int_a^b \left\| \frac{d}{dt} \mathbf{x}\!\left( \sum_{i=0}^{n} \mathbf{b}_i
    N_{i,d}(t)\right) \right\|^2 dt,
  \end{equation}
  where $\mathbf{b}_i = (u_i,v_i) \in U $ for each $i=0,1, \dots, n$ denotes the control point of the  B-spline curve.

  \begin{defn}\label{BsplineGeodesicPQ}
  Let $\mathbf{x}(u,v)$ be a parametrization of a parametric surface
  $\Sigma$ where $\mathbf{x} : U \subset \mathbb{R}^2 \rightarrow \Sigma$.
  We define the space of B-spline curve of degree $d$ on the clamped knot vector $K$
  in equation~\eqref{eqn_knot_vector}, with control points in $U$ as
  \begin{equation}
  \mathcal{V} = \left\{ \gamma(t) = \sum_{i=0}^n N_{i,d}(t)(u_i,v_i) \;\middle|\; (u_i,v_i) \in U \right\}.
  \end{equation}

  Given two points $p,q \in U$, a curve $\tilde{\gamma}(t)\subset U$ is
  called a B-spline geodesic-like curve of degree $d$ on $\Sigma$
  connecting $p$ and $q$ if
  \begin{equation}
  \tilde{\gamma} \in \mathcal{V}, \qquad
  \tilde{\gamma}(a)=p,\; \tilde{\gamma}(b)=q
  \end{equation}
  and satisfies the optimality condition
  \begin{equation}
  \nabla E(\mathbf{b}) = 0,
  \end{equation}
  where
  $\mathbf{b}=(u_1,\dots,u_{n-1},v_1,\dots,v_{n-1})$
  denotes the vector of free control points, and
  \begin{equation}\label{energyofBsplineGeodesic_likePQ}
E(\gamma)= \frac{1}{2} \int_a^b \left\| \frac{d}{dt}
\mathbf{x}(\gamma(t)) \right\|^2 dt
\end{equation}
  is the energy functional of the curve $\gamma\in {\mathcal V}$
  satisfying $\gamma(a)={\bf p}$ and $\gamma(b)={\bf q}$.
  \end{defn}

  It is clear that a minimal B-spline geodesic-like curve is also a B-spline geodesic-like curve.

  Since the derivative of the B-spline curve $\gamma (t)$ is
  itself a B-spline curve of degree $d-1$, it can be expressed as
  $$
  \gamma'(t) = \sum_{i=0}^n N_{i,d}'(t) \mathbf{b}_i
  = \sum_{i=0}^{n-1} N_{i, d-1}(t) \Delta \mathbf{b}_i$$
  where the derivative of the basis functions satisfies
  $$ N_{i,d}'(t) = \frac{d}{t_{i+d} - t_i} N_{i,d-1}(t)
  - \frac{d}{t_{i+d+1} - t_{i+1}}  N_{i+1,d-1}(t)$$
  and the first-order control differences
  $\Delta \mathbf{b}_i \in \mathbb{R}^m$ are given by
  $$\Delta \mathbf{b}_i = \frac{d}{t_{i+d+1} - t_{i+1}}
  (\mathbf{b}_{i+1} - \mathbf{b}_i).$$

  We present the energy function $E$ in equation
  (\ref{energyofBsplineGeodesic_likePQ}), together with
  its gradient and Hessian.
  Let $\mathbf{x}:U \subset \mathbb{R}^2 \rightarrow \mathbb{R}^3$
  be a parametrization of a parametric surface $\Sigma$, and let
  $\gamma(t) = \sum_{i=0}^n N_{i,d}(t)(u_i,v_i)$ be a B-spline
  curve on $U$.
  For simplicity, we write $\mathbf{x}$ to denote
  $\mathbf{x}(\gamma(t))$ and
  $\gamma$ denote $\gamma(t)$. The notation $(\cdot)'$
  denotes differentiation
  with respect to the parameter $t$.

  The energy of $\gamma$ is defined as the function
  \begin{equation}
    \begin{aligned}
  E(u_1,\dots,u_{n-1},v_1,\dots,v_{n-1})
  &= \frac{1}{2}\int_a^b \langle \mathbf{x}' , \mathbf{x}' \rangle
  dt\\
     &= \frac{1}{2} \int_a^b \langle \mathbf{x}_u u' + \mathbf{x}_v v', \mathbf{x}_u u' + \mathbf{x}_v v' \rangle  dt.
  \end{aligned}
  \end{equation}
  Here, we set $\gamma(t) = (u(t),v(t))$ where
  $$u(t) = \sum_{i=0}^n N_{i,d}(t)u_i,~
  v(t) = \sum_{i=0}^n N_{i,d}(t)v_i.$$
  Since
  \begin{equation}
    \frac{\partial u(t)}{\partial u_i} = N_{i,d}(t), \quad
    \frac{\partial u(t)}{\partial v_i} = 0, \quad \frac{\partial
    v(t)}{\partial u_i} = 0, \quad \frac{\partial v(t)}{\partial v_i} =
    N_{i,d}(t),
  \end{equation}
  it follows from the chain rule that
  \begin{equation}\label{xtu_i}
  \frac{\partial \mathbf{x}'}{\partial u_i} =
  (\mathbf{x}_{uu}u' + \mathbf{x}_{uv}v')N_{i,d}(t) + \mathbf{x}_u N_{i,d}'(t),
  \end{equation}
  and
  \begin{equation}\label{xtv_i}
  \frac{\partial \mathbf{x}'}{\partial v_i}  =
  (\mathbf{x}_{uv}u' + \mathbf{x}_{vv}v')N_{i,d}(t) + \mathbf{x}_v N_{i,d}'(t).
  \end{equation}
  Hence, we obtain
\begin{equation}
\begin{aligned}
    E_{u_i} &= \int_a^b \left\langle (\mathbf{x}_{uu}u' +
\mathbf{x}_{uv}v')N_{i,d}(t)+ \mathbf{x}_u N'_{i,d}(t),
\mathbf{x}_u u' + \mathbf{x}_v v' \right\rangle dt, \\
    E{v_i} &= \int_a^b \left\langle (\mathbf{x}_{uv}u' +
\mathbf{x}_{vv}v')N_{i,d}(t)+ \mathbf{x}_v N'_{i,d}(t), \mathbf{x}_u
u' + \mathbf{x}_v v' \right\rangle dt.
\end{aligned}
\end{equation}
  The gradient of the energy functional is given by
  \begin{equation}
  \nabla E = (E_{u_1},\dots,E_{u_{n-1}},E_{v_1},\dots,E_{v_{n-1}}).
  \end{equation}
  Therefore, a B-spline geodesic-like curve is obtained as a critical
point of $E$, i.e., it satisfies
\begin{equation}
\begin{aligned}
&\int_a^b \left\langle (\mathbf{x}_{uu}u' +
\mathbf{x}_{uv}v')N_{i,d}(t)+ \mathbf{x}_u N'_{i,d}(t), \mathbf{x}_u
u'+ \mathbf{x}_v v' \right\rangle dt = 0, \\
&\int_a^b \left\langle (\mathbf{x}_{uv}u' +
\mathbf{x}_{vv}v')N_{i,d}(t)+\mathbf{x}_v N'_{i,d}(t), \mathbf{x}_u
u' + \mathbf{x}_v v' \right\rangle dt = 0,
\end{aligned}
\end{equation}
for $i=1,\ldots,n-1$.

  Similarly, the entries of the Hessian matrix of the energy functional are given by
\begin{equation}
\begin{aligned}
E_{u_i u_j} &= \int_a^b \left( \left\langle \frac{\partial^2
\mathbf{x}'}{\partial u_i \partial u_j}, \mathbf{x}' \right\rangle +
\left\langle \frac{\partial \mathbf{x}'}{\partial u_i},
\frac{\partial \mathbf{x}'}{\partial u_j} \right\rangle \right)
dt,\\
E_{u_i v_j} &= \int_a^b \left( \left\langle \frac{\partial^2
\mathbf{x}'}{\partial u_i \partial v_j}, \mathbf{x}' \right\rangle +
\left\langle \frac{\partial \mathbf{x}'}{\partial u_i},
\frac{\partial \mathbf{x}'}{\partial v_j} \right\rangle \right)
dt,\\
E_{v_i v_j} &= \int_a^b \left( \left\langle \frac{\partial^2
\mathbf{x}'}{\partial v_i \partial v_j}, \mathbf{x}' \right\rangle +
\left\langle \frac{\partial \mathbf{x}'}{\partial v_i},
\frac{\partial \mathbf{x}'}{\partial v_j} \right\rangle \right) dt.
\end{aligned}
\end{equation}
  Consequently, the Hessian of E can be written in block matrix form
as
  \begin{equation}
  \begin{pmatrix}
  E_{uu} & E_{uv} \\
  E_{vu} & E_{vv}
  \end{pmatrix},
  \end{equation}
  where $E_{uu} = (E_{u_i u_j})$, $E_{uv}=E_{vu}=(E_{u_i v_j})$, and $E_{vv}=(E_{v_i v_j})$.

  Since the energy functional $E$, together with its gradient and Hessian, are smooth functions
  of the B-spline control point variables,
  the geodesic computation reduces to a finite-dimensional nonlinear optimization problem that
  can be efficiently solved using Newton-type or quasi-Newton methods~\cite{Chen}.

Although geodesic-like curves provide an elegant framework for
approximating geodesic boundary value problems, their construction
remains computationally challenging. Classical numerical schemes,
such as Newton-type methods, can in principle be applied to solve
the resulting highly nonlinear optimality conditions associated with
the energy functional; however, these methods are often
computationally expensive and may suffer from numerical instability
in practice.


\section{Neural Network Frameworks for Minimal B-spline Geodesic-like Curves}

Since deep learning provides a robust alternative that alleviates issues related to computational instability and 
initialization dependence, we develop a neural network-based framework for computing minimal B-spline geodesic-like 
curves on parametric surfaces.
   
Let $\Sigma$ be a regular surface given by a parametrization $\mathbf{x}(u, v): U \to \mathbb{R}^3$, where the 
parameter domain is assumed to be a rectangular region $U = (a, b) \times (c, d) \subseteq \mathbb{R}^2$.To find the 
minimal B-spline geodesic-like curve connecting two points $p$ and $q$ on $\Sigma$, we construct a deep fully connected 
neural network $\mathcal{N}_\theta$. The objective is to predict a set of interior control points 
$\mathbf{c} = [u_1, \dots, u_{n-1}, v_1, \dots, v_{n-1}]^T$ such that the energy functional $E$ of the 
corresponding B-spline curve is minimized.

\begin{figure}
{\centering
\includegraphics[width=\textwidth]{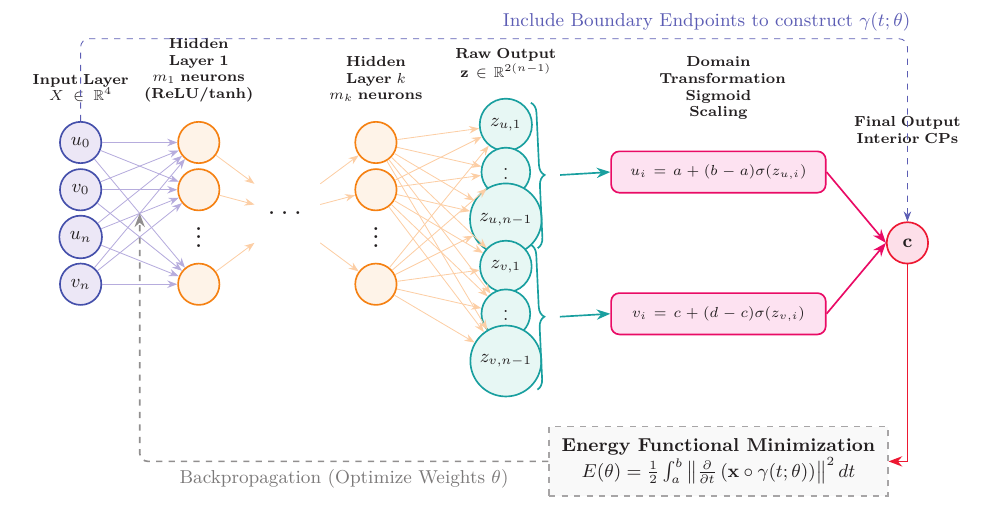}
\caption{The architecture of the point-specific optimization model.}\label{fig_NN_model1} 
}
\end{figure}

The feed-forward structure of the network, the Basic Model,  is defined as follows, see Figure~\ref{fig_NN_model1} 
for an illustration:

\begin{enumerate}
    \item Input Layer: 
    The input vector $X \in \mathbb{R}^4$ consists of the coordinates 
    of the two endpoints in the parameter domain $U$:
    \begin{equation}
        X = [u_0, v_0, u_n, v_n]^T.
    \end{equation}

    \item Sequence of Deep Hidden Layers: 
    The network comprises $k$ hidden layers to extract high-dimensional 
    non-linear geometric features. The architecture is defined as:
    \begin{itemize}
        \item Hidden Layer 1: Contains $m_1$ neurons with 
        ReLU or tanh activation functions.
        \item Hidden Layer 2: Contains $m_2$ neurons.
        \item Intermediate Layers ($\dots$): Represents 
        additional hidden layers as required.
        \item Hidden Layer $k$: Contains $m_k$ neurons, 
        serving as the final feature extraction layer.
    \end{itemize}

    \item Raw Output Layer: 
    The network generates raw predicted values $\mathbf{z} \in 
    \mathbb{R}^{2(n-1)}$:
    \begin{equation}
        \mathbf{z} = [z_{u,1}, \dots, z_{u,n-1}, z_{v,1}, 
        \dots, z_{v,n-1}]^T.
    \end{equation}

    \item Domain Transformation: 
    To ensure the predicted control points strictly reside within 
    the domain $U = (a, b) \times (c, d)$, a Sigmoid-based 
    scaling, $\sigma(x) = \frac{1}{1 + e^{-x}}$, is applied:
    \begin{equation}
    \left\{
    \begin{array}{ll}
        u_i = a + (b - a)\sigma(z_{u,i}), \\
        v_i = c + (d - c)\sigma(z_{v,i}),
    \end{array}
    \right.
    \end{equation}
    for each $i = 1, \dots, n-1$.

    \item Final Output: 
    The final output vector $\mathbf{c}$ defines the interior 
    control points used to construct the variational B-spline 
    curve $\gamma(t; \theta)$.
\end{enumerate}

The entire network is optimized by minimizing the energy 
functional $E(\theta)$:
\begin{equation}
    E(\theta) = \frac{1}{2} \int_a^b \left\| \frac{\partial }{\partial t} 
    \left( \mathbf{x} \circ \gamma(t; \theta) \right) 
    \right\|^2 dt.
\end{equation}

By minimizing $E$ via backpropagation, the network parameters 
$\theta$ are optimized to yield the optimal control points 
$\mathbf{c}^* = \mathcal{N}_{\theta^*}(X)$, defining the 
minimal energy B-spline geodesic-like curve on $\Sigma$.

\begin{figure}
\centering
\includegraphics[width=\textwidth]{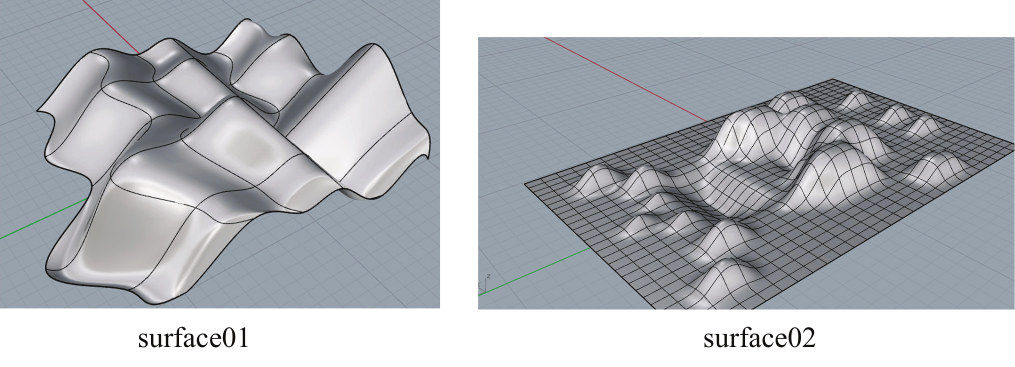}
\caption{Two Bspline surfaces of degree $3 \times 3$. }\label{fig_surface0102}
\end{figure}

\begin{figure}
\centering
\includegraphics[width=\textwidth]{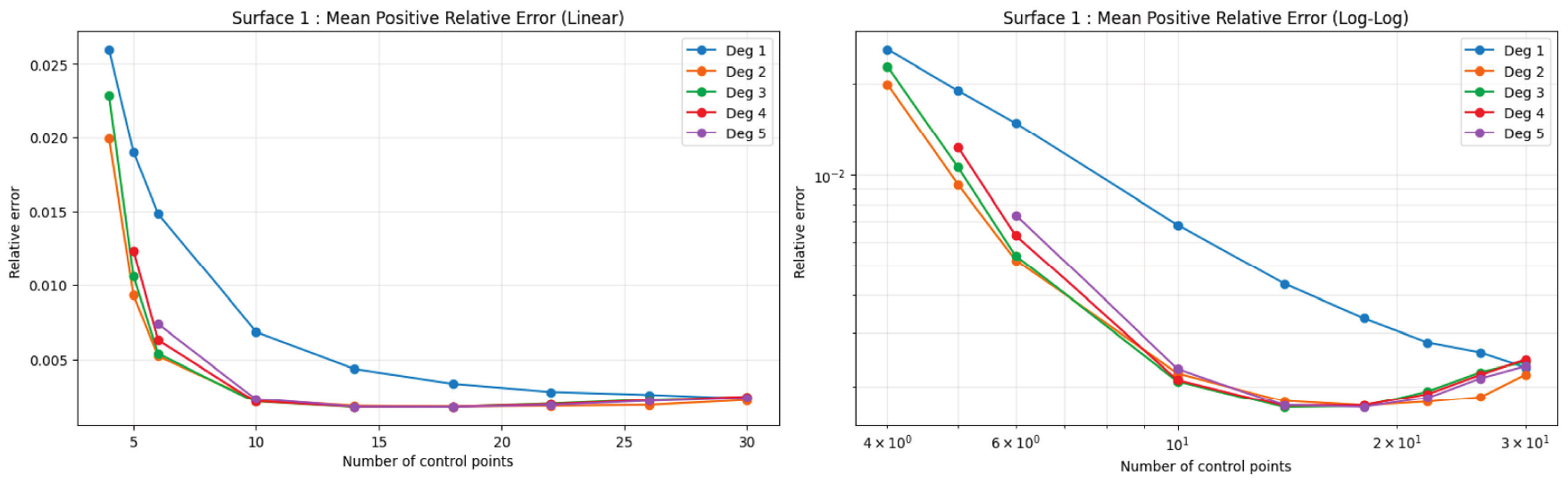}
\caption{Mean relative errors of geodesic-like curves on surface .}\label{fig_rel_err_mean_surface1}
\end{figure}

\begin{figure}
\centering
\includegraphics[width=\textwidth]{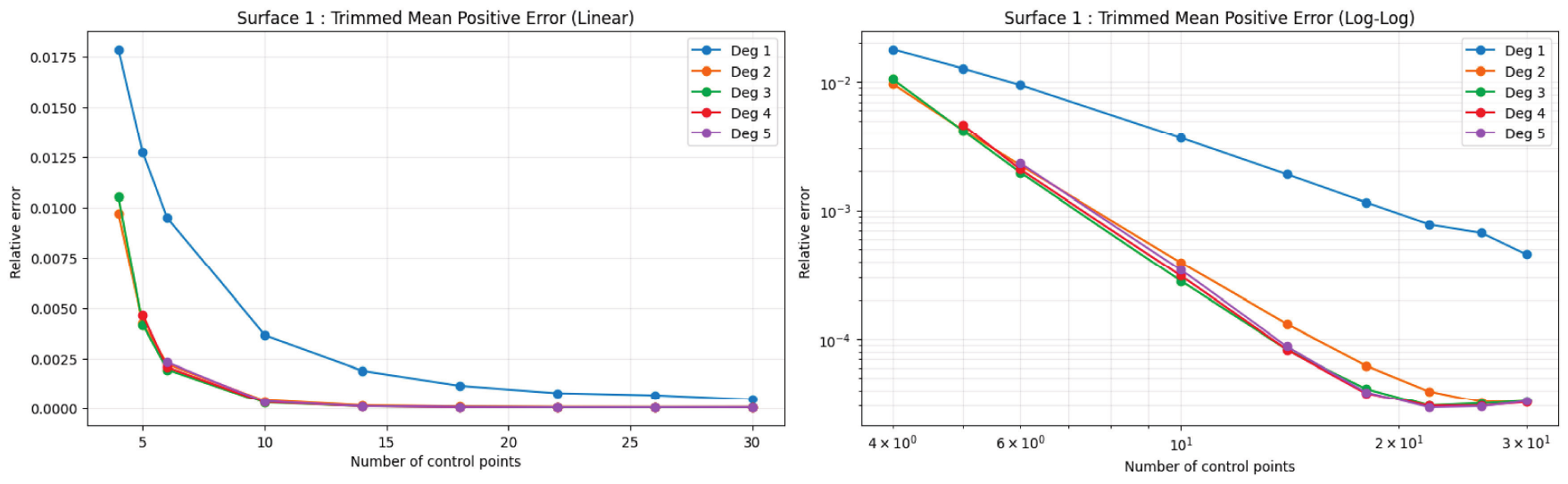}
\caption{Trimmed mean relative errors( $ 5\% - 95 \% $ ) of geodesic-like curves on surface 1.}\label{fig_rel_err_trimmed_mean_surface1}
\end{figure}

\begin{figure}
\centering
\includegraphics[width=\textwidth]{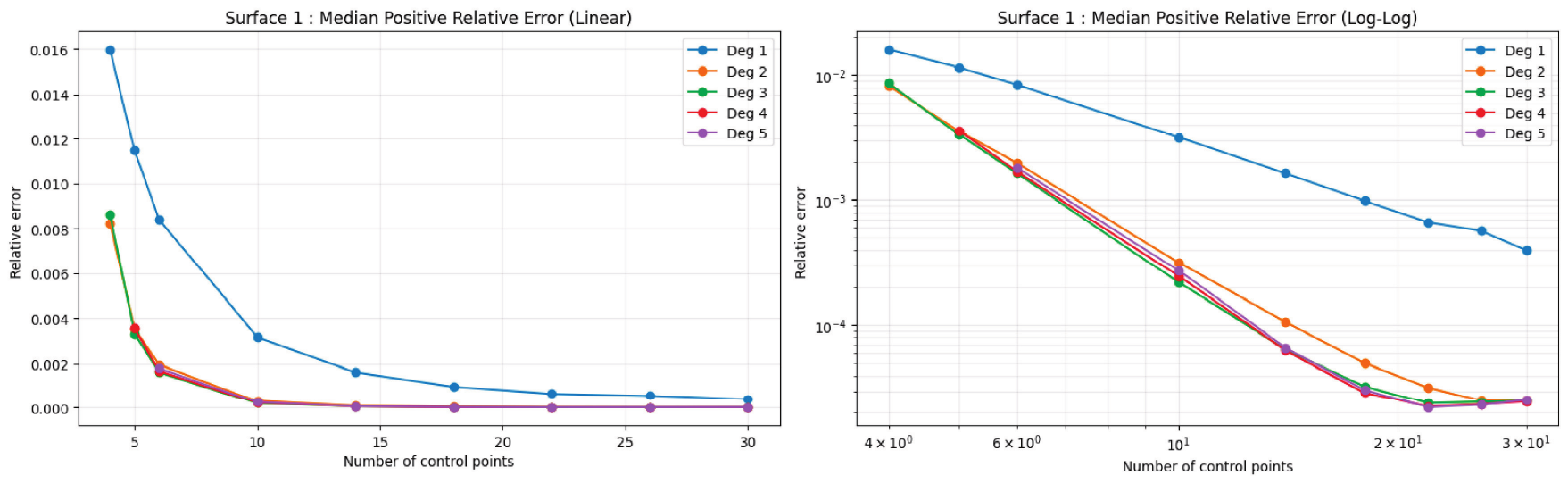}
\caption{Median relative errors of geodesic-like curves on surface 1.}\label{fig_rel_err_median_surface1}
\end{figure}

\begin{table}[htbp]
\centering
\scriptsize
\tabcolsep=4pt
\caption{Mean relative errors on surface 1}\label{table_rel_err_mean_surface1}
\begin{tabular}{|c||c|c|c|c|c|c|c|c|c|}
\hline
Degree & \multicolumn{9}{c|}{Numbers of Control Points} \\ \cline{2-10}
 & 4 & 5 & 6 & 10 & 14 & 18 & 22 & 26 & 30 \\ \hline
1 & 2.60e-02 & 1.90e-02 & 1.48e-02 & 6.83e-03 & 4.36e-03 & 3.36e-03 & 2.80e-03 & 2.60e-03 & 2.32e-03 \\ \hline
2 & 2.00e-02 & 9.28e-03 & 5.22e-03 & 2.22e-03 & 1.79e-03 & 1.72e-03 & 1.78e-03 & 1.85e-03 & 2.20e-03 \\ \hline
3 & 2.29e-02 & 1.06e-02 & 5.39e-03 & 2.09e-03 & 1.70e-03 & 1.71e-03 & 1.93e-03 & 2.24e-03 & 2.42e-03 \\ \hline
4 & - & 1.23e-02 & 6.30e-03 & 2.11e-03 & 1.72e-03 & 1.73e-03 & 1.89e-03 & 2.19e-03 & 2.46e-03 \\ \hline
5 & - & - & 7.37e-03 & 2.29e-03 & 1.73e-03 & 1.70e-03 & 1.84e-03 & 2.13e-03 & 2.34e-03 \\ \hline
\end{tabular}
\end{table}

\begin{table}[htbp]
\centering
\scriptsize
\tabcolsep=4pt
\caption{Trimmed mean relative errors on surface 1}\label{table_rel_err_trimmed_mean_surface1}
\begin{tabular}{|c||c|c|c|c|c|c|c|c|c|}
\hline
Degree & \multicolumn{9}{c|}{Numbers of Control Points} \\ \cline{2-10}
 & 4 & 5 & 6 & 10 & 14 & 18 & 22 & 26 & 30 \\ \hline
1 & 1.79e-02 & 1.27e-02 & 9.48e-03 & 3.65e-03 & 1.90e-03 & 1.15e-03 & 7.80e-04 & 6.69e-04 & 4.54e-04 \\ \hline
2 & 9.69e-03 & 4.23e-03 & 2.23e-03 & 3.90e-04 & 1.31e-04 & 6.30e-05 & 3.90e-05 & 3.20e-05 & 3.20e-05 \\ \hline
3 & 1.05e-02 & 4.17e-03 & 1.96e-03 & 2.81e-04 & 8.30e-05 & 4.10e-05 & 3.00e-05 & 3.10e-05 & 3.30e-05 \\ \hline
4 & - & 4.62e-03 & 2.07e-03 & 3.08e-04 & 8.30e-05 & 3.80e-05 & 3.00e-05 & 3.00e-05 & 3.20e-05 \\ \hline
5 & - & - & 2.32e-03 & 3.43e-04 & 8.80e-05 & 3.90e-05 & 2.90e-05 & 3.00e-05 & 3.30e-05 \\ \hline
\end{tabular}
\end{table}

\begin{table}[htbp]
\centering
\scriptsize
\tabcolsep=4pt
\caption{Median relative errors on surface 1}\label{table_rel_err_median_surface1}
\begin{tabular}{|c||c|c|c|c|c|c|c|c|c|}
\hline
Degree & \multicolumn{9}{c|}{Numbers of Control Points} \\ \cline{2-10}
 & 4 & 5 & 6 & 10 & 14 & 18 & 22 & 26 & 30 \\ \hline
1 & 1.60e-02 & 1.15e-02 & 8.40e-03 & 3.17e-03 & 1.63e-03 & 9.81e-04 & 6.64e-04 & 5.69e-04 & 3.96e-04 \\ \hline
2 & 8.24e-03 & 3.54e-03 & 1.97e-03 & 3.15e-04 & 1.06e-04 & 5.00e-05 & 3.20e-05 & 2.50e-05 & 2.50e-05 \\ \hline
3 & 8.63e-03 & 3.32e-03 & 1.63e-03 & 2.19e-04 & 6.50e-05 & 3.30e-05 & 2.40e-05 & 2.40e-05 & 2.50e-05 \\ \hline
4 & - & 3.57e-03 & 1.68e-03 & 2.46e-04 & 6.30e-05 & 2.90e-05 & 2.20e-05 & 2.30e-05 & 2.40e-05 \\ \hline
5 & - & - & 1.80e-03 & 2.72e-04 & 6.70e-05 & 3.10e-05 & 2.20e-05 & 2.30e-05 & 2.50e-05 \\ \hline
\end{tabular}
\end{table}

\begin{figure}
\centering
\includegraphics[width=\textwidth]{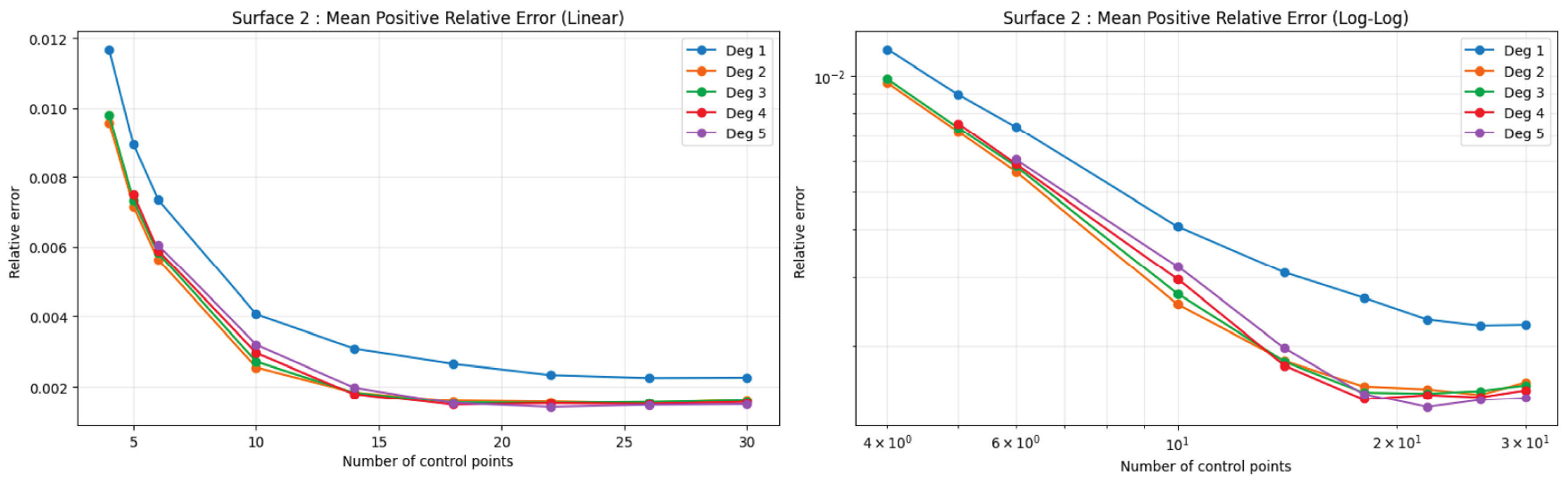}
\caption{Mean relative errors of geodesic-like curves on surface 2.}\label{fig_rel_err_mean_surface2}
\end{figure}

\begin{figure}
\centering
\includegraphics[width=\textwidth]{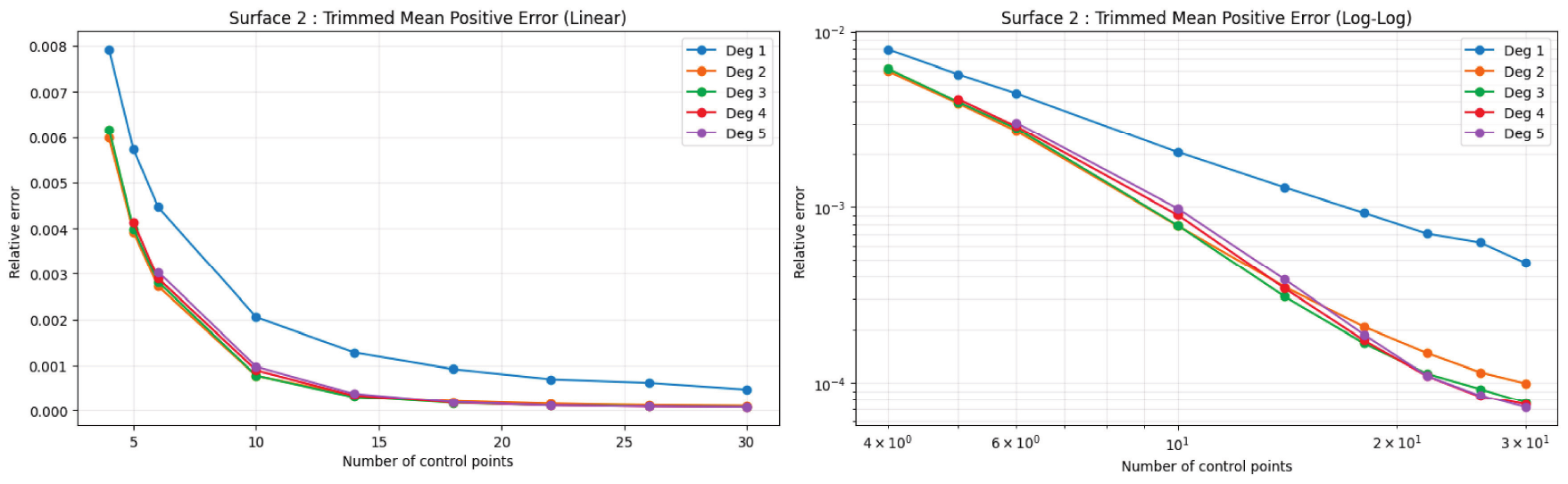}
\caption{Trimmed mean relative errors( $ 5\% - 95 \% $ ) of geodesic-like curves on surface 2.}\label{fig_rel_err_trimmed_mean_surface2}
\end{figure}

\begin{figure}
\centering
\includegraphics[width=\textwidth]{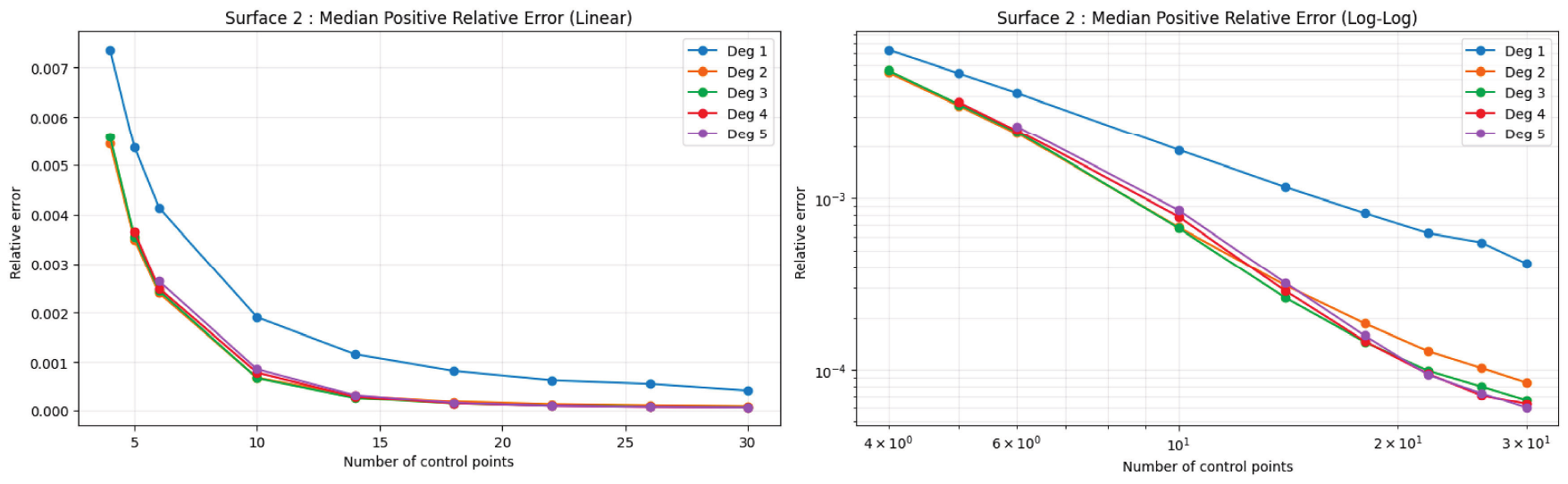}
\caption{Median relative errors of geodesic-like curves on surface 2..}\label{fig_rel_err_median_surface2}
\end{figure}

\begin{table}[htbp]
\centering
\scriptsize
\tabcolsep=4pt
\caption{Mean  relative errors on surface 2}\label{table_rel_err_mean_surface2}
\begin{tabular}{|c|c|c|c|c|c|c|c|c|c|}
\hline
Degree & \multicolumn{9}{c|}{Numbers of Control Points} \\ \cline{2-10}
 & 4 & 5 & 6 & 10 & 14 & 18 & 22 & 26 & 30 \\ \hline
1 & 1.17e-02 & 8.94e-03 & 7.36e-03 & 4.06e-03 & 3.09e-03 & 2.65e-03 & 2.33e-03 & 2.25e-03 & 2.26e-03 \\ \hline
2 & 9.58e-03 & 7.16e-03 & 5.63e-03 & 2.55e-03 & 1.83e-03 & 1.57e-03 & 1.54e-03 & 1.49e-03 & 1.61e-03 \\ \hline
3 & 9.80e-03 & 7.33e-03 & 5.82e-03 & 2.72e-03 & 1.82e-03 & 1.51e-03 & 1.50e-03 & 1.53e-03 & 1.58e-03 \\ \hline
4 & -        & 7.51e-03 & 5.89e-03 & 2.96e-03 & 1.78e-03 & 1.45e-03 & 1.49e-03 & 1.47e-03 & 1.53e-03 \\ \hline
5 & -        & -        & 6.05e-03 & 3.20e-03 & 1.97e-03 & 1.50e-03 & 1.38e-03 & 1.44e-03 & 1.46e-03 \\ \hline
\end{tabular}
\end{table}

\begin{table}[htbp]
\centering
\scriptsize
\tabcolsep=4pt
\caption{Trimmed Mean relative error( $ 5\% - 9 \%$)  on surface 2}\label{table_rel_err_trimmed_mean_surface2}
\begin{tabular}{|c|c|c|c|c|c|c|c|c|c|}
\hline
Degree & \multicolumn{9}{c|}{Numbers of Control Points} \\ \cline{2-10}
 & 4 & 5 & 6 & 10 & 14 & 18 & 22 & 26 & 30 \\ \hline
1 & 7.92e-03 & 5.72e-03 & 4.47e-03 & 2.05e-03 & 1.29e-03 & 9.27e-04 & 7.06e-04 & 6.29e-04 & 4.79e-04 \\ \hline
2 & 5.98e-03 & 3.92e-03 & 2.72e-03 & 7.83e-04 & 3.50e-04 & 2.08e-04 & 1.48e-04 & 1.15e-04 & 1.00e-04 \\ \hline
3 & 6.17e-03 & 3.99e-03 & 2.82e-03 & 7.88e-04 & 3.08e-04 & 1.68e-04 & 1.13e-04 & 9.20e-05 & 7.70e-05 \\ \hline
4 & -        & 4.14e-03 & 2.90e-03 & 9.03e-04 & 3.43e-04 & 1.74e-04 & 1.09e-04 & 8.30e-05 & 7.50e-05 \\ \hline
5 & -        & -        & 3.04e-03 & 9.83e-04 & 3.86e-04 & 1.88e-04 & 1.10e-04 & 8.50e-05 & 7.20e-05 \\ \hline
\end{tabular}
\end{table}

\begin{table}[htbp]
\centering
\scriptsize
\tabcolsep=4pt
\caption{Median  relative errors on surface 2}\label{table_rel_err_median_surface2}
\begin{tabular}{|c|c|c|c|c|c|c|c|c|c|}
\hline
Degree & \multicolumn{9}{c|}{Numbers of Control Points} \\ \cline{2-10}
 & 4 & 5 & 6 & 10 & 14 & 18 & 22 & 26 & 30 \\ \hline
1 & 7.36e-03 & 5.36e-03 & 4.14e-03 & 1.91e-03 & 1.16e-03 & 8.21e-04 & 6.29e-04 & 5.56e-04 & 4.19e-04 \\ \hline
2 & 5.44e-03 & 3.49e-03 & 2.40e-03 & 6.80e-04 & 3.12e-04 & 1.87e-04 & 1.29e-04 & 1.03e-04 & 8.50e-05 \\ \hline
3 & 5.58e-03 & 3.55e-03 & 2.45e-03 & 6.73e-04 & 2.64e-04 & 1.46e-04 & 9.90e-05 & 8.00e-05 & 6.70e-05 \\ \hline
4 & -        & 3.66e-03 & 2.50e-03 & 7.84e-04 & 2.90e-04 & 1.47e-04 & 9.50e-05 & 7.10e-05 & 6.30e-05 \\ \hline
5 & -        & -        & 2.64e-03 & 8.54e-04 & 3.23e-04 & 1.59e-04 & 9.40e-05 & 7.30e-05 & 6.00e-05 \\ \hline
\end{tabular}
\end{table}

We compute the numerically minimal B-spline geodesic-like curves on  two different surfaces (see Figure \ref{fig_surface0102}) with
$5,000$ random pairs of endpoints on the domains of these two surfaces. For the following simulations, we apply a specific neural 
network model which consists of two hidden layers with 32 and 64 neurons and  set the training tolerance
 to $1.0 \times 10^{-4}$.

Since analytical solutions for general geodesics on arbitrary surfaces are typically unavailable, 
we establish a benchmark for performance evaluation by synthesizing four independent numerical approaches. 
To determine the reference ground-truth length ($L_{\text{exact}}$), we compute the geodesic length using each of the following 
methods  and adopt the smallest value among them as the best available approximation to the true geodesic length:
\begin{enumerate}
    \item The numerical method proposed by Zhang \cite{Zhang1}; 
    \item The fourth-order Runge-Kutta (RK4) method utilizing an extremely fine step size; 
    \item The built-in shortest path computation provided by the \texttt{Rhino3D 8} software;
    \item The geodesic optimization tools available within the \texttt{Grasshopper} environment.
\end{enumerate}

We define the relative error as:
\begin{equation}
    \text{rel\_err} = \frac{L_{\text{approx}} - L_{\text{exact}}}{L_{\text{exact}}},
\end{equation}
where $L_{\text{approx}}$ and $L_{\text{exact}}$ denote the approximated and reference lengths, respectively. 
Since the reference lengths obtained from existing methods may not coincide with the absolute shortest paths, 
we do not take the absolute value of the relative error. 
A negative relative error ($\text{rel\_err} < 0$) indicates that our approach achieves a shorter path length than 
the baseline methods, and thus our analysis focuses primarily on the comparison of positive errors.

Figures \ref{fig_rel_err_mean_surface1}-\ref{fig_rel_err_median_surface1} and Tables 
\ref{table_rel_err_mean_surface1}-\ref{table_rel_err_median_surface1} present the mean, trimmed mean ($5\%-95\%$), 
and median relative errors for the estimated results on surface 1. Corresponding results for surface 2 are present in 
Figures \ref{fig_rel_err_mean_surface2}-\ref{fig_rel_err_median_surface2} and Tables 
\ref{table_rel_err_mean_surface2}-\ref{table_rel_err_median_surface2}.

\begin{figure}
\centering

\includegraphics[width=.45\textwidth]{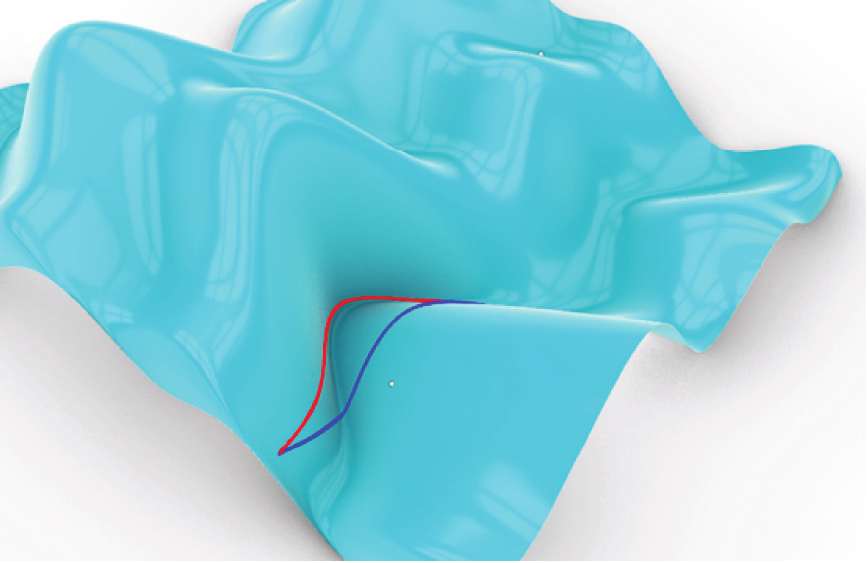}
\includegraphics[width=.45\textwidth]{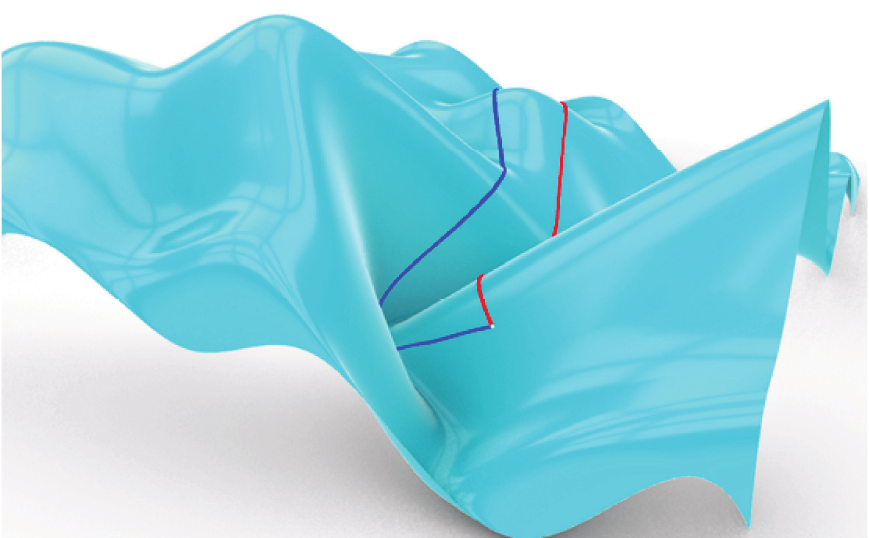}
\caption{Challenges in Global Optimization. Left subfigure is a failure case of our proposed method; 
`right subfigure is a failure case of baseline methods.}\label{fig_global_optimal_problrem}
\end{figure}

\subsection*{Observations on Numerical Performance}

The experimental results provide insights into the effectiveness and
limitations of our approach. In particular, we observe three key
phenomena regarding its performance and convergence behavior.

\begin{enumerate}

    \item \textbf{Convergence and Model Saturation}:

    The relative error consistently decreases as the number of control
    points increases. However, a saturation effect is observed,
    whereby the accuracy ceases to improve beyond a certain threshold, approximately
    $1.0\times 10^{-3}$.
    This phenomenon can be attributed to the limitations of the model architecture.
    Specifically, the model employs two hidden layers with 32 and 64 neurons, respectively, and the training
    tolerance is set to $1.0\times 10^{-4}$. These design choices constrain the expressive
    capacity of the model, such that increasing the degrees of freedom no longer improves
    accuracy once the error reaches this limit.

    \item \textbf{Geometric Efficiency of Basis Functions}:

    Increasing the polynomial degree beyond 2 provides only
    marginal improvement. This behavior is expected because the approximation
    is performed in the parameter plane, where the torsion of plane curves is identically zero.
    Consequently, a quadratic B-spline representation is
    sufficient to capture the geometric properties of curves in the parameter plane.
    Higher polynomial degrees do not yield significant performance gains.
    Therefore, degree 2 provides the best balance between computational complexity and approximation accuracy.

\item \textbf{Challenges in Global Optimization}:

    The difference between the Mean and Trimmed Mean relative errors
    reveals two key challenges in the proposed geodesic computation approach.

    The first challenge is the presence of local minima. Significant
    deviations (relative errors $>0.1$) often occur when the
    optimization process becomes trapped in such minima.

   Figure \ref{fig_global_optimal_problrem} illustrates two representative cases of this phenomenon.
   The red curves represent the B-spline geodesic-like curves generated by the proposed method,
   while the blue curves correspond to the baseline geodesic approximations.
   In the left subfigure, the proposed method converges to a local minimum with a path length of 13.869,
   whereas the baseline method identifies a superior solution with a length of 12.389.
   Conversely, in the right subfigure, the baseline method becomes trapped in a local minimum,
   yielding a path length of 19.532, while the proposed approach identifies a more optimal path with a length of 17.202.
   These examples highlight that both the proposed method and the baseline techniques are susceptible to local minima,
   and neither consistently guarantees convergence to the global optimum.

    The second challenge arises from boundary constraints. When the
    true shortest path lies near the domain boundary, restricting
    control points within the domain prevents the model from
    approaching the global optimum; see Figure \ref{fig_boundary_restict1}.
    This reflects an inherent geometric limitation rather than an optimization failure.

    These two difficulties (on-convex energy landscapes and boundary
    constraints)
    constitute persistent challenges for all numerical approaches,
    thereby motivating the proposed neural network strategy.

\end{enumerate}

\begin{figure}
\centering
\includegraphics[width=.5\textwidth]{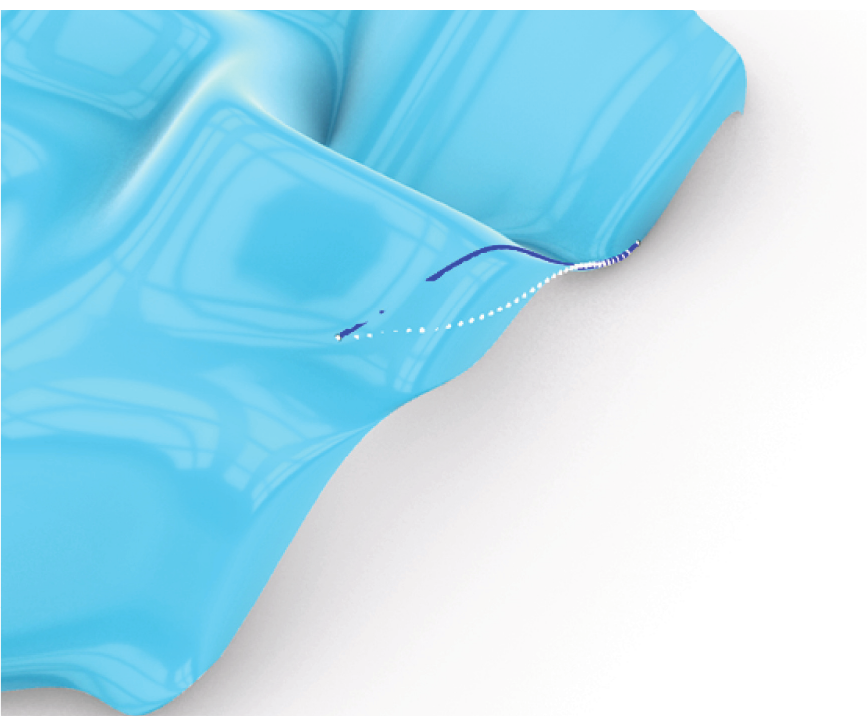}
\caption{Boundary limitation of the domain-restricted method. A portion of the true shortest path (white dashed line) follows
the domain boundary, while our neural-approximated curve (blue line) deviates from this boundary due to the
 domain constraint.}\label{fig_boundary_restict1}
\end{figure}

\begin{figure}
\centering
\includegraphics[width=\textwidth]{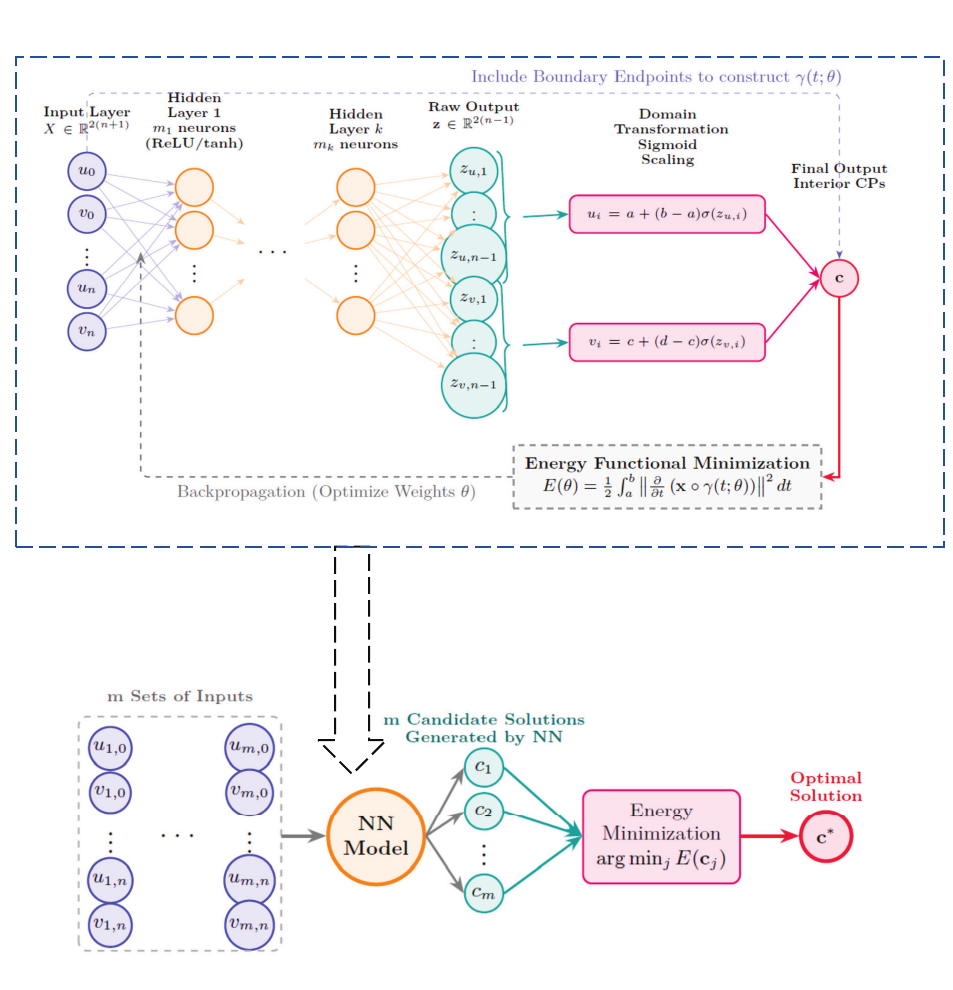}
\caption{Search mechanism of the revised model using $m$ perturbed initial curves to identify the minimal B-spline geodesic-like
 curve.}\label{fig_revised_model}
\end{figure}

To address the issue of becoming trapped in local minima, we adopt a
multi-path search strategy. We modify the input of the basic model
by incorporating all control points of an initial guess curve,
thereby expanding the input dimension to $2(n+1)$ values, rather
than restricting it to the four endpoint coordinates. We refer to
this configuration as the Revised Model.

Inspired by variational principles, we generate initial paths
through linear interpolation between the endpoints to obtain $n-1$
intermediate points, which serve as the control points of a straight
B-spline curve. By introducing controlled perturbations to the
internal control points while keeping the endpoints fixed, we
generate $m$ distinct initial paths. Each set of control points is
subsequently fed into the Revised Model, yielding $m$ candidate
B-spline curves. Finally, a selection mechanism evaluates these
candidates and identifies the one with the minimum energy (path
length) as the estimated shortest path; see Figure
\ref{fig_revised_model}.

 To generate diverse initial paths, we employ a bell-shaped perturbation strategy applied to the internal control
  points of the initial B-spline curve. Let the initial straight path be defined by  $n+1$ control points
  $\mathbf{p}_i = (u_i, v_i)$ for $i=0, \dots, n$. The endpoints $\mathbf{p}_0$ and $\mathbf{p}_n$
  are kept fixed to satisfy the boundary conditions.
   For the internal control points $\mathbf{p}_i$ where $i=1, \dots, n-1$,  a displacement vector
   is applied along the geometric normal of the straight path and modulated by a bell-shaped weight function
   $$w_i = \sin(\frac{i\pi}{n}).$$
   This weighting ensures that the perturbation magnitude vanishes at the boundaries
   and reaches its maximum at the center of the curve.
   By scaling the perturbation with a coefficient $\alpha$,
   we generate $m$ variations of the initial path,
   thereby effectively sampling the parameter space surrounding the straight-line
   trajectory and reducing sensitivity to initialization.

The efficacy of the Revised Model is illustrated in Figure
\ref{fig_local_optim}. The left subfigure presents a scenario in
which the Basic Model, restricted to endpoint inputs, becomes
trapped in a suboptimal local minimum (red curve, length = 14.872),
deviating significantly from the baseline geodesic (blue curve,
length = 14.311). The right subfigure demonstrates the performance
of the Revised Model under the proposed perturbation strategy with
coefficients $\alpha \in \{-0.3, 0.0, 0.3\}$. Specifically, the path
generated with $\alpha=0.0$ (red curve, length = 14.872) yields the
same result as the Basic Model, confirming that the Revised Model
recovers the original prediction when no perturbation is applied.
Furthermore, the path corresponding to $\alpha = 0.3$ (blue curve,
length = 14.311) reproduces the baseline geodesic. Most importantly,
the path corresponding to $\alpha=-0.3$ (green curve, length =
13.476) escapes the local minimum identified by the Basic Model.
Geometrically, this green curve aligns more closely with the true
shortest-path direction on the surface, confirming that the proposed
initialization strategy, combined with multi-path evaluation,
enables the model to navigate the energy landscape more effectively
and converge toward the global shortest geodesic.

\begin{figure}
\centering
\includegraphics[width=.45\textwidth]{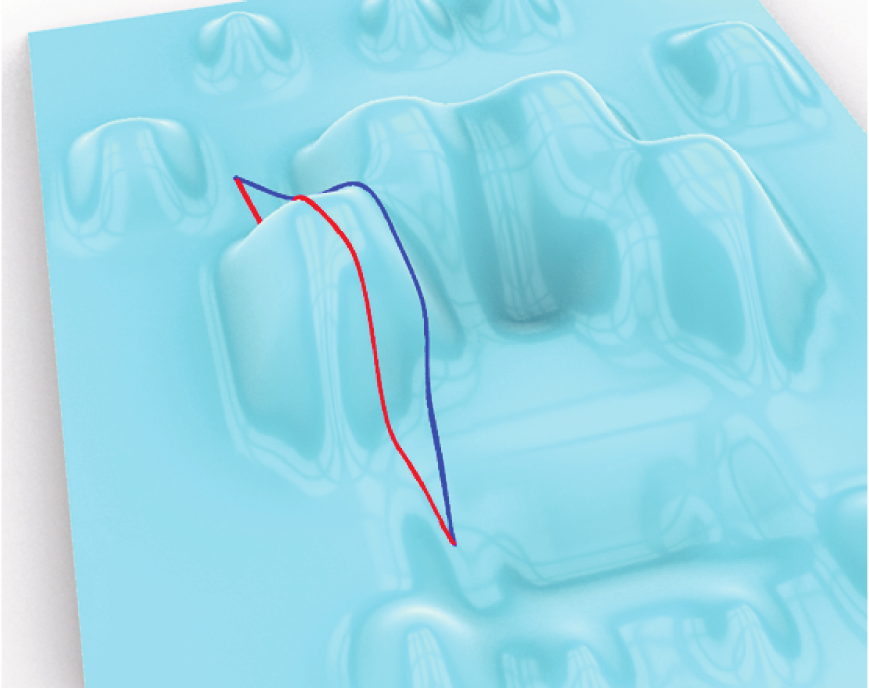}
\includegraphics[width=.45\textwidth]{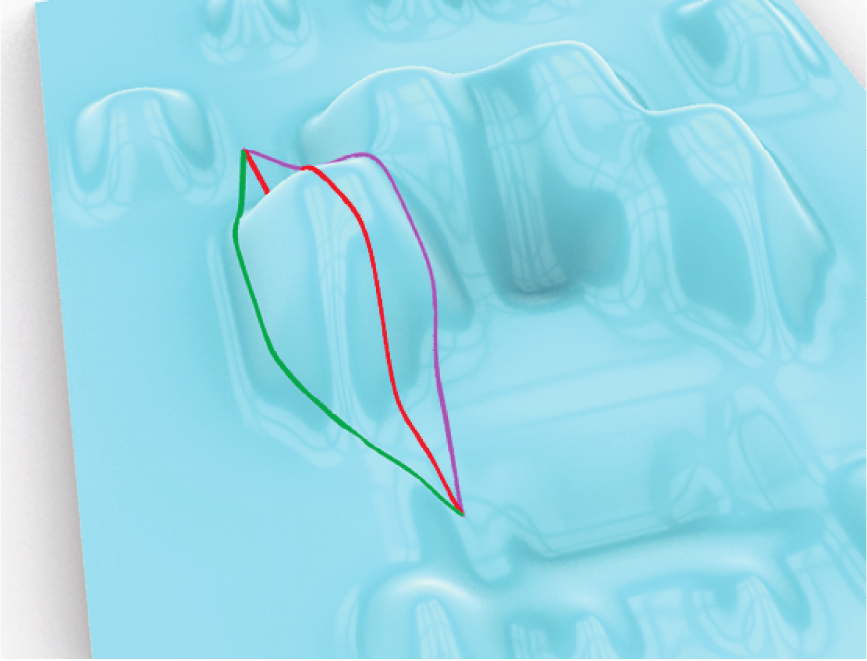}
\caption{Left:Comparison between the approximated geodesic curve derived from the basic model and 
the baseline geodesic. Rithg: Comparison of the revised model with three different perturbation coefficients.}\label{fig_local_optim}
\end{figure}

To address the issue of control points moving outside the parameter
domain, it is necessary to extend the domain range. We consider the
case of two surfaces joined along a common edge with at least $C^0$
continuity. In our framework, surface information is incorporated
only through the loss function, allowing us to focus exclusively on
transformations of the parametric domains. Specifically, a global
parametrization is achieved by mapping the domain of the first
surface to $[0,1]\times[0,1]$, with the shared boundary located at
$u=1$. The domain of the second surface is then mapped to
$[1,2]\times[0,1]$, such that its corresponding shared boundary is
precisely aligned with the same interface at $u=1$. By adjusting the
orientation of the second surface to ensure consistency across the
shared interface, the global domain is extended from $[0,1]\times[0,1]$
to $[0,2]\times[0,1]$. Figure \ref{fig_multisurf01} illustrates an
example of two surfaces connected with $C^1$ continuity. This
approach can also be generalized to accommodate various
configurations in which multiple surfaces are joined along their
boundaries.

To establish a unified framework for geometric stitching across
multiple surface patches, we introduce a coordinate mapping function
 $\phi:\Omega\subset \mathbb{R}^2 \to [0,1]\times [0,1]$. The proposed mapping effectively handles
discrepancies in domain dimensions, vertex alignments, and chiral
orientations among adjacent parametric domains.

Let the initial domain $\Omega \subset \mathbb{R}^2$ be a rectangle defined by its four vertices
in counter-clockwise order: $\mathbf{q}_1(x_1, y_1)$, $\mathbf{q}_2(x_2, y_1)$, $\mathbf{q}_3(x_2, y_2)$, and $\mathbf{q}_4(x_1, y_2)$.
The target domain is the standard unit square $\hat{\Omega} = [0,1]\times[0,1]$ with vertices
$\mathbf{p}_1(0,0)$, $\mathbf{p}_2(1,0)$, $\mathbf{p}_3(1,1)$, and $\mathbf{p}_4(0,1)$.

To flexibly enforce the boundary constraint $\phi(\mathbf{q}_1) = \mathbf{p}_i$ for $i \in \{1, 2, 3, 4\}$ while accommodating
both orientation-preserving and orientation-reversing parameterizations, we parameterize the mapping $\phi$
via homogeneous coordinates:

To flexibly enforce the boundary constraint $\phi(\mathbf{q}_1) = \mathbf{p}_i$ for $i
\in \{1, 2, 3, 4\}$, while accommodating both orientation-preserving
and orientation-reversing parameterizations, we parameterize the
mapping $\phi$ using homogeneous coordinates:
\begin{equation}
\begin{pmatrix} u \cr v \cr 1 \end{pmatrix} = M_{\text{total}}(i, \sigma) \begin{pmatrix} x \cr y \cr 1 \end{pmatrix}
\end{equation}

where $(x,y) \in \Omega$ represents coordinates in the original
domain, and $(u,v) \in \hat{\Omega}$ denotes the transformed
coordinates in the computational domain. The overall transformation
matrix $M_{\text{total}} \in \mathbb{R}^{3 \times 3}$ is elegantly
constructed through the factorization of three discrete geometric
operators:

\begin{equation}
M_{\text{total}}(i, \sigma) = R(\theta_i) \circ F(\sigma) \circ
M_{\text{norm}}
\end{equation}
The constituent matrices are defined as follows:

\paragraph*{1. Domain Normalization Matrix ($M_{\text{norm}}$):}
This operator maps the arbitrary rectangle $\Omega$ onto the standard unit square $\hat{\Omega}$
via scaling and translation, aligning $\mathbf{q}_1$ directly with $\mathbf{p}_1$:

This operator maps the arbitrary rectangular domain $\Omega$ onto
the standard unit square $\hat{\Omega}$ via scaling and translation,
aligning $\mathbf{q}_1$ directly with $\mathbf{p}_1$:
\begin{equation}
M_{\text{norm}} =
\begin{pmatrix} \frac{1}{x_2 - x_1} & 0 & -\frac{x_1}{x_2 - x_1} \cr
 0 & \frac{1}{y_2 - y_1} & -\frac{y_1}{y_2 - y_1} \cr 0 & 0 & 1 \end{pmatrix}
\end{equation}

\paragraph*{2. Chiral Reflection Operator ($F(\sigma)$):}
To reconcile opposite vertex-wrapping sequences (chiral orders)
between adjacent patches, the binary parameter $\sigma \in \{1,-1\}$
activates either an identity mapping or a coordinate permutation
across the diagonal $u=v$:
\begin{equation}
F(\sigma) =
\begin{pmatrix} \frac{1+\sigma}{2} & \frac{1-\sigma}{2} & 0 \cr
\frac{1-\sigma}{2} & \frac{1+\sigma}{2} & 0 \cr 0 & 0 & 1 \end{pmatrix}
\end{equation}

\paragraph*{3. Discrete Center-Rotation Matrix ($R(\theta_i)$):}
Finally, to map the primary anchor point $\mathbf{q}_1$ from $\mathbf{p}_1$ to any
designated target corner $\mathbf{p}_i$, a discrete rotation about the center
of the unit square $(0.5,0.5)$ is applied:
\begin{equation}
R(\theta_i) =
\begin{pmatrix} \cos\theta_i & -\sin\theta_i & \frac{1}{2}(1-\cos\theta_i+\sin\theta_i) \cr
\sin\theta_i & \cos\theta_i & \frac{1}{2}(1-\cos\theta_i-\sin\theta_i) \cr 0 & 0 & 1 \end{pmatrix}
\end{equation}
The rotation angle $\theta_i$ is determined strictly by the target vertex index $i$:
\begin{equation}
\theta_i = \frac{\pi}{2}(i - 1), \quad \text{for } i \in \{1, 2, 3, 4\}
\end{equation}

  We now address the issue of control points being constrained within the parameter domain. By extending the boundary
  edges of adjacent surfaces, we construct an extended surface  that preserve at least $C^0$ continuity
  with the original surfaces. This process allows the control points, which are previously constrained near the
  domain boundaries, to be relocated into the interior of a unified global parametrization.
  Consequently, this approach enables the generation of geodesic-like curves that traverse along the original surface boundaries. 
  Figure \ref{fig_boundary_restict2}  effectively resolves the domain restriction problem encountered in 
  Figure \ref{fig_boundary_restict1}. Figure \ref{fig_boundary_restict2} presents the geodesic-like curve approximation 
  within the global parameter domain. The transparent surface on the right is an auxiliary surface
generated by a straight extrusion of the boundary curve of the left surface, preserving $C^0$ continuity across the 
interface, as shown in the left panel. 

   The right panel displays the corresponding geodesic-like
   curves on the actual surface geometry. The blue curve, generated using the local domain approach, has a
   length of 11.391, whereas the red curve, obtained using the proposed global domain method, has a length of 11.058.
   As evident from the right panel, the use of a global parametrization allows the control points
   to extend beyond the original surface boundaries.
   Since the shortest path between two boundary points remains strictly confined to the boundary,
   the optimization naturally yields a curve that aligns with the interface rather than entering the extruded surface region.
   As shown in the right panel, global parametrization permits control
points to extend beyond the original surface boundaries. Since the
shortest path between two boundary points is constrained to lie on
the boundary, the optimization naturally produces a curve that
adheres to the interface rather than entering the extruded region.

\begin{figure}
\centering
\includegraphics[width=.45\textwidth]{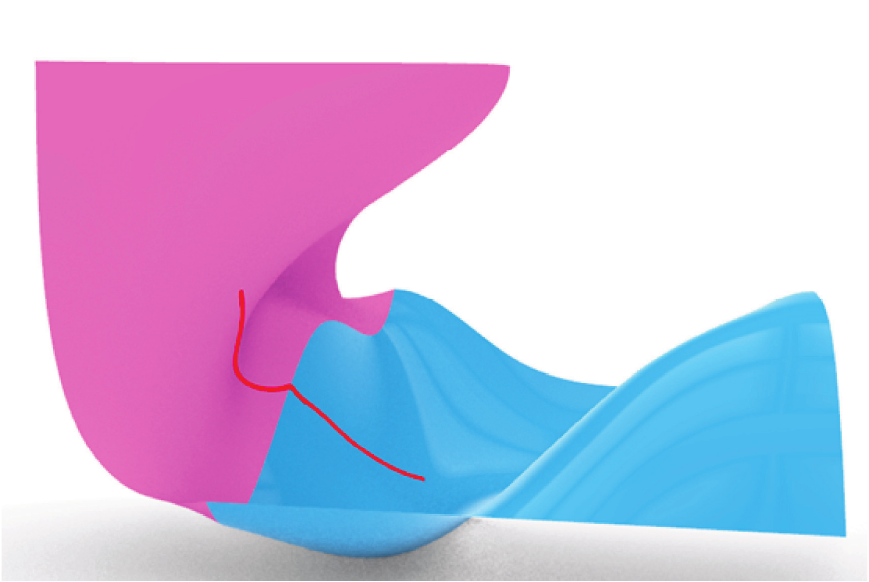}
\includegraphics[width=.45\textwidth]{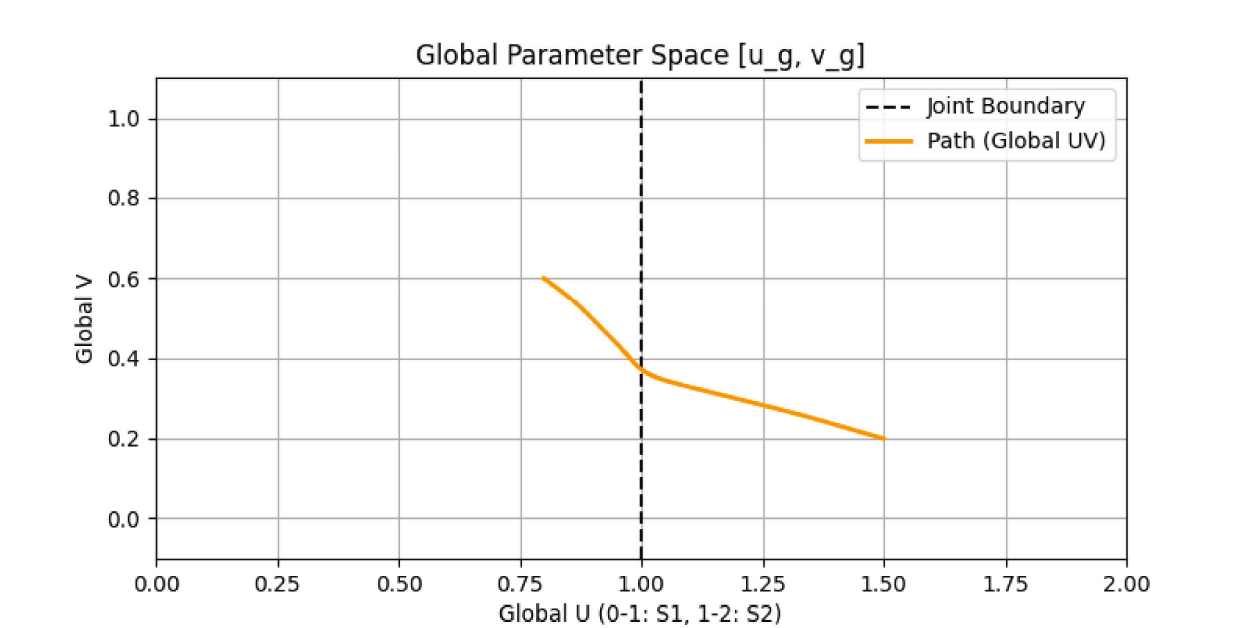}
\caption{Approximation of a geodesic-like curve spanning two surfaces. The right panel shows the curve in the parametric domain,
while the red curve in the left panel represents the geodesic-like approximation crossing both surfaces.}\label{fig_multisurf01}
\end{figure}

\begin{figure}
\centering
\includegraphics[width=.45\textwidth]{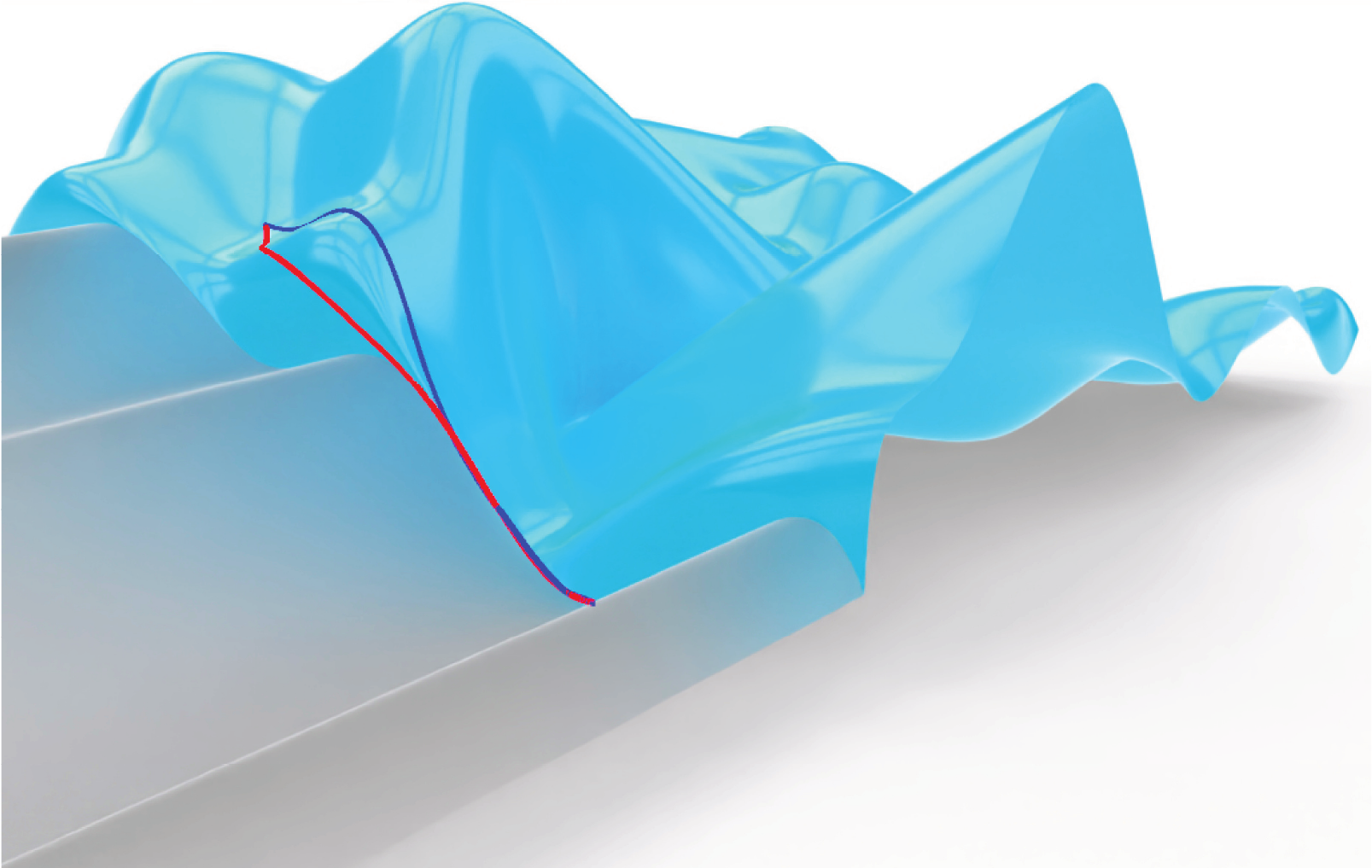}
\includegraphics[width=.45\textwidth]{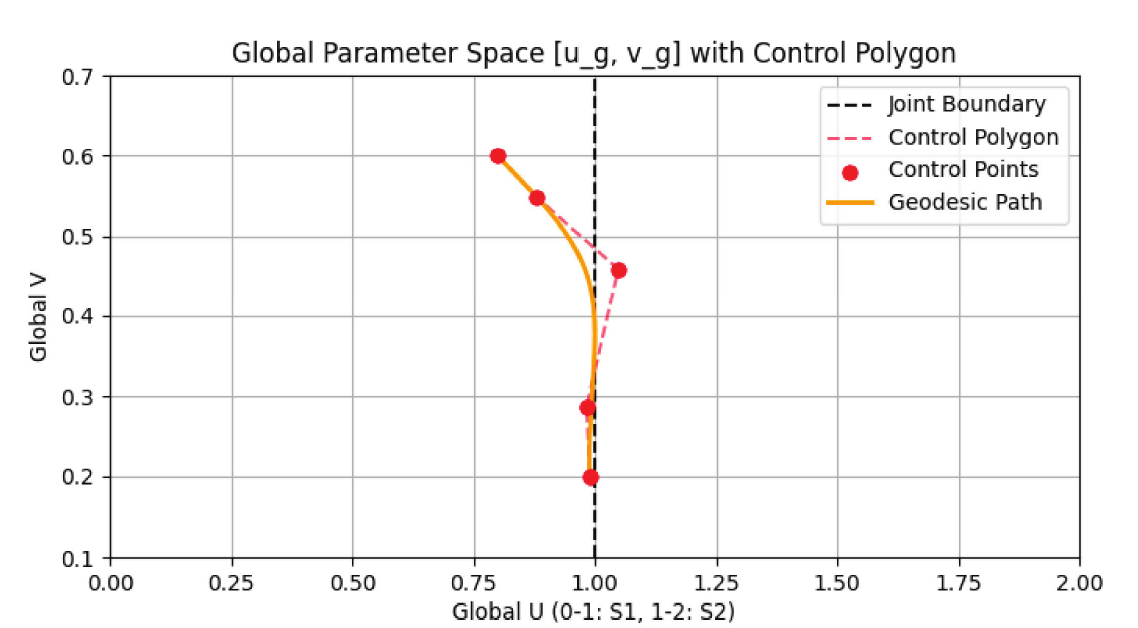}
\caption{Geodesic-like curve approximation using the global domain method:
The right panel illustrates the curves in the global parametric domain, where the light blue surface represents the original
surface and the yellow surface represents the extruded extension. The blue curve (length: 11.391) depicts the geodesic-like
approximation obtained using a local domain approach, while the red curve (length: 11.058) represents the result obtained via
the proposed global domain method.}\label{fig_boundary_restict2}
\end{figure}

By modifying the parameter domain, we can effectively address the
challenges posed by self-intersecting regular (parametric) surfaces.
Consider a surface $\Sigma$ with a parameterization $\mathbf{x}: [a,
b] \times [c, d] \to \Sigma$ such that $\mathbf{x}(a, t) =
\mathbf{x}(b, t)$ for all $t \in [c, d]$. This implies that the
boundaries at $u=a$ and $u=b$ coincide. We define a mapping $\phi:
\mathbb{R} \times [c, d] \to [a, b] \times [c, d]$ by $\phi(u, v) =
(\bar{u}, v)$, where $\bar{u} \in [a, b]$ is given by $u = \bar{u} +
n(b-a)$ for some $n \in \mathbb{Z}$. Consequently, we define a new
function $\bar{\mathbf{x}} = \mathbf{x} \circ \phi$ with the domain
$\mathbb{R} \times [c, d]$, whose image remains identical to
$\Sigma$. For convenience, we refer to this parameterization as a
u-lattice parameterization and its domain as a u-lattice domain over
$[a, b] \times [c, d]$. Evidently, $\bar{\mathbf{x}}(u, v) =
\bar{\mathbf{x}}(u + n(b-a), v)$ for every integer $n$. From this
perspective, we can compute geodesic-like curves on surfaces with
self-intersecting boundaries. Figure \ref{fig_geodesic_revolsurface}
illustrates geodesic-like estimations on three surfaces of
revolution: the red curve represents the shortest path between two
points, the purple curve represents a geodesic-like curve winding
once around the surface, and the blue curve represents a
geodesic-like curve winding twice around the surface, respectively.

\begin{figure}
\centering
\includegraphics[width=.3\textwidth]{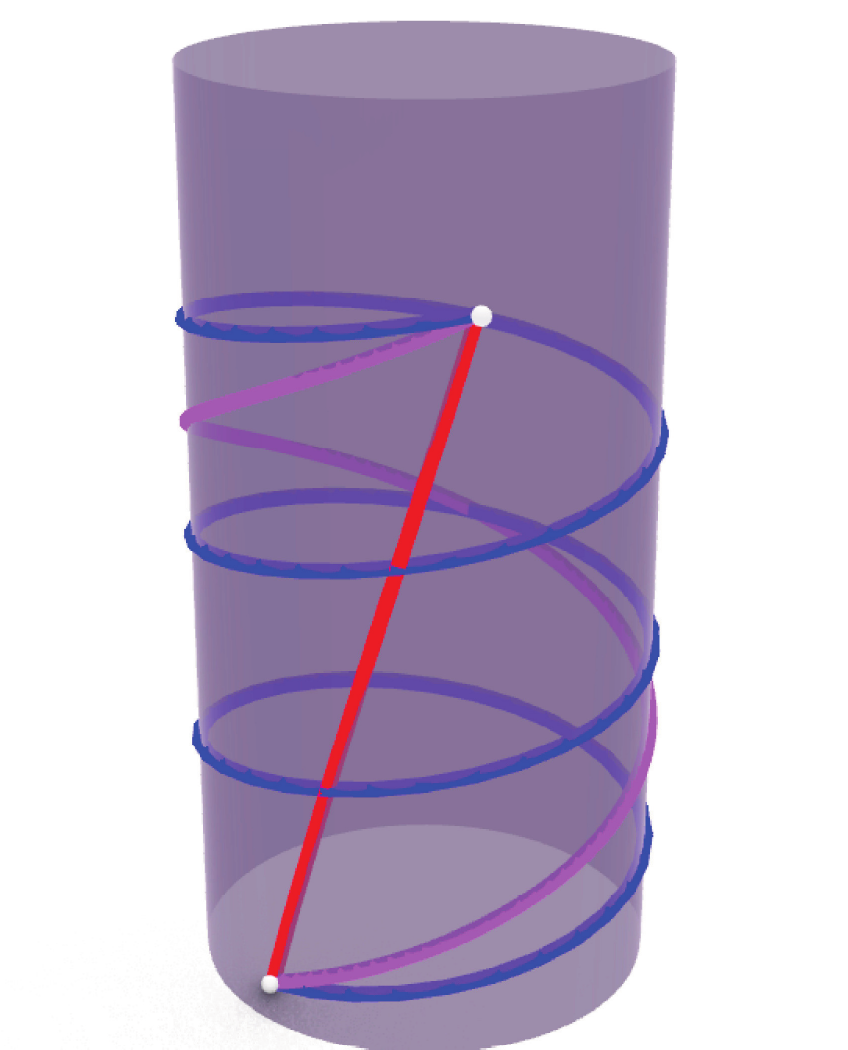}
\includegraphics[width=.3\textwidth]{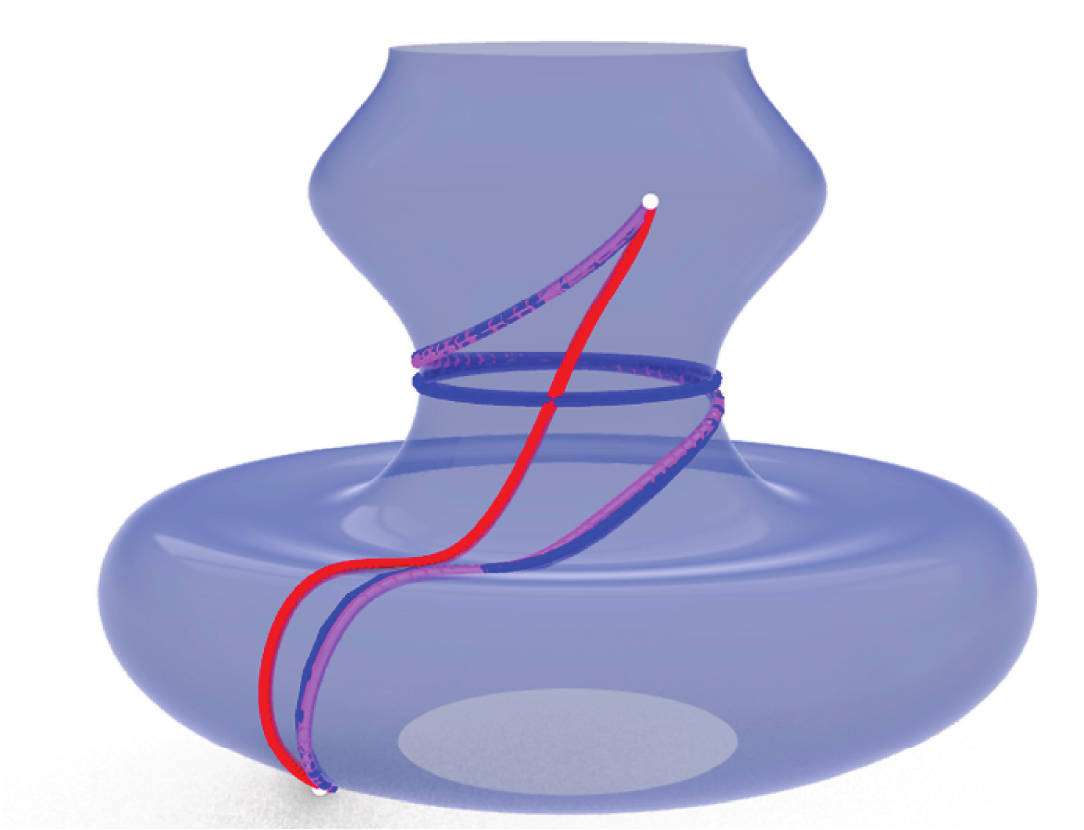}
\includegraphics[width=.3\textwidth]{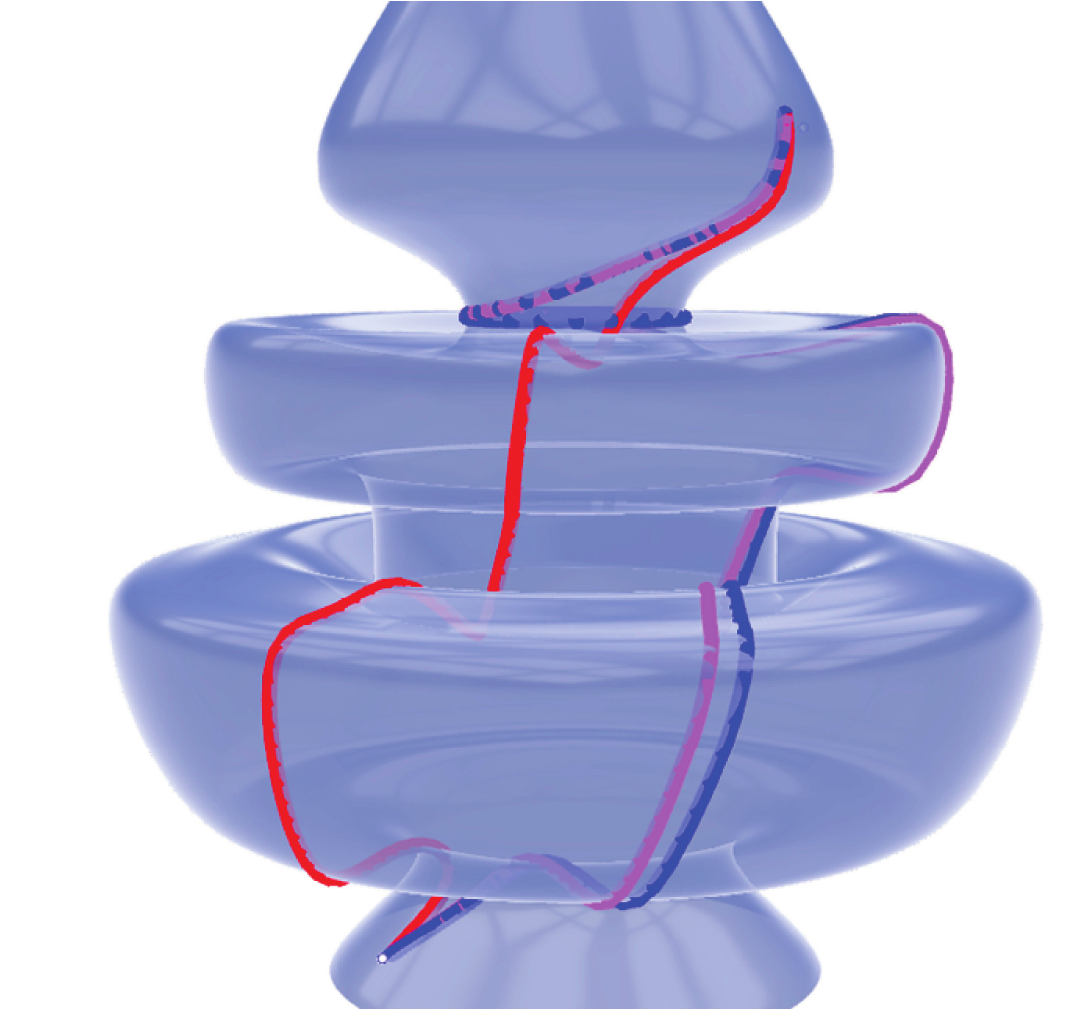}
\caption{Approximations of geodesic-like curves on surface of revolutions}\label{fig_geodesic_revolsurface}
\end{figure}

\section{Conclusion}

In this paper, we presented a novel neural-network-based framework
for geodesic-like curve estimation on parametric surfaces. The main
findings of this study are summarized as follows:

\begin{enumerate}
    \item \textbf{Computational Efficiency and Scalability:}
     We demonstrated that, for surfaces with moderate geometric variations,
     even a lightweight fully connected neural network, consisting of two hidden layers with 32 and 64 neurons,
     can effectively approximate geodesic-like trajectories. The estimated path lengths are competitive
     with those produced by established numerical methods. Furthermore, the proposed framework exhibits strong scalability:
     for more complex surface geometries, the network depth and width can be increased, while residual connections
     may be incorporated to mitigate gradient vanishing and thereby enhance robustness.

    \item \textbf{Robustness Against Local Optima:}
    A significant challenge in geodesic computation is the tendency
    of traditional solvers to become trapped in local minima. By employing a multi-path search strategy within
    our neural network framework, Xessentially a parallel optimization mechanism, we achieve improved global
    exploration capability, thereby providing a more efficient alternative to conventional iterative schemes.

    \item \textbf{Unified Global Parametrization for Multi-surface Systems:}
    By extending the concept of parameter-domain manipulation, we successfully addressed the estimation of geodesic-like curves
    across interconnected multi-surface systems. This approach effectively resolves boundary-related constraints
    by allowing control points to extend beyond the original local domains.
    The introduction of an extended surface configuration facilitates smooth transitions across parametric domains,
    thereby overcoming the computational difficulties commonly encountered when geodesics approach surface interfaces.

    \item \textbf{Topological Flexibility and Generalization:}
    Our proposed domain-manipulation strategy provides a versatile framework capable of handling diverse
    surface topologies, including surfaces of revolution and closed surfaces.
    By leveraging lattice parameterization, the framework effectively computes not only the shortest path
    between two points but also higher-order winding paths, such as curves wrapping around the surface $n$ times.
    These results confirm the broad applicability of the proposed approach across a wide range of topological
    and geometric settings.
\end{enumerate}

\clearpage 
\appendix
\section{Appendix}

\begin{table}[H]
\centering
\small 
\tabcolsep=2pt 
\caption{The relative errors of surface 1}
\begin{tabular}{|c|c|c|c|c|c|c|c|c|c|c|c|}
\hline
Deg & Cpts & Mean & Median & Trim. Mean & P95 & IQR & IQR Ratio & Neg & $>0.01$ & $>0.02$ & $>0.10$ \\ \hline
1 & 4 & 0.025968 & 0.015985 & 0.017857 & 0.084826 & 0.031950 & 0.146989 & 37 & 3072 & 2220 & 149 \\ \hline
1 & 5 & 0.018999 & 0.011464 & 0.012714 & 0.062532 & 0.023337 & 0.130554 & 45 & 2638 & 1719 & 44 \\ \hline
1 & 6 & 0.014830 & 0.008399 & 0.009484 & 0.047769 & 0.018564 & 0.104042 & 50 & 2246 & 1319 & 29 \\ \hline
1 & 10 & 0.006831 & 0.003170 & 0.003654 & 0.022490 & 0.007516 & 0.047188 & 75 & 1047 & 304 & 9 \\ \hline
1 & 14 & 0.004359 & 0.001626 & 0.001899 & 0.012907 & 0.004047 & 0.025918 & 84 & 388 & 162 & 10 \\ \hline
1 & 18 & 0.003355 & 0.000981 & 0.001153 & 0.008572 & 0.002469 & 0.015979 & 96 & 215 & 145 & 10 \\ \hline
1 & 22 & 0.002800 & 0.000664 & 0.000780 & 0.006447 & 0.001649 & 0.010732 & 101 & 182 & 144 & 11 \\ \hline
1 & 26 & 0.002596 & 0.000569 & 0.000669 & 0.005860 & 0.001384 & 0.009026 & 111 & 176 & 142 & 10 \\ \hline
1 & 30 & 0.002320 & 0.000396 & 0.000454 & 0.004407 & 0.000902 & 0.005902 & 117 & 170 & 145 & 12 \\ \hline
2 & 4 & 0.019955 & 0.008240 & 0.009686 & 0.084624 & 0.022111 & 0.086415 & 43 & 2223 & 1459 & 170 \\ \hline
2 & 5 & 0.009281 & 0.003539 & 0.004226 & 0.037863 & 0.009277 & 0.052263 & 56 & 1262 & 587 & 31 \\ \hline
2 & 6 & 0.005216 & 0.001968 & 0.002234 & 0.017769 & 0.004816 & 0.033551 & 80 & 629 & 218 & 7 \\ \hline
2 & 10 & 0.002218 & 0.000315 & 0.000390 & 0.005717 & 0.000975 & 0.007356 & 123 & 173 & 135 & 5 \\ \hline
2 & 14 & 0.001787 & 0.000106 & 0.000131 & 0.003173 & 0.000323 & 0.002445 & 176 & 162 & 131 & 6 \\ \hline
2 & 18 & 0.001724 & 0.000050 & 0.000063 & 0.002277 & 0.000146 & 0.000961 & 253 & 163 & 133 & 7 \\ \hline
2 & 22 & 0.001776 & 0.000032 & 0.000039 & 0.002752 & 0.000097 & 0.000636 & 473 & 159 & 133 & 9 \\ \hline
2 & 26 & 0.001845 & 0.000025 & 0.000032 & 0.003112 & 0.000083 & 0.000546 & 799 & 157 & 131 & 7 \\ \hline
2 & 30 & 0.002196 & 0.000025 & 0.000032 & 0.006248 & 0.000088 & 0.000577 & 1187 & 169 & 138 & 10 \\ \hline
3 & 4 & 0.022859 & 0.008631 & 0.010514 & 0.099715 & 0.026747 & 0.104847 & 38 & 2307 & 1594 & 246 \\ \hline
3 & 5 & 0.010569 & 0.003324 & 0.004168 & 0.048208 & 0.010470 & 0.054071 & 60 & 1364 & 737 & 47 \\ \hline
3 & 6 & 0.005391 & 0.001626 & 0.001962 & 0.021663 & 0.004784 & 0.035879 & 90 & 666 & 267 & 9 \\ \hline
3 & 10 & 0.002087 & 0.000219 & 0.000281 & 0.005394 & 0.000733 & 0.005554 & 162 & 178 & 135 & 5 \\ \hline
3 & 14 & 0.001702 & 0.000065 & 0.000083 & 0.003169 & 0.000222 & 0.001679 & 282 & 155 & 128 & 6 \\ \hline
3 & 18 & 0.001714 & 0.000033 & 0.000041 & 0.002497 & 0.000106 & 0.000806 & 494 & 158 & 130 & 6 \\ \hline
3 & 22 & 0.001930 & 0.000024 & 0.000030 & 0.003756 & 0.000080 & 0.000528 & 943 & 161 & 133 & 7 \\ \hline
3 & 26 & 0.002235 & 0.000024 & 0.000031 & 0.006793 & 0.000085 & 0.000560 & 1509 & 160 & 132 & 8 \\ \hline
3 & 30 & 0.002422 & 0.000025 & 0.000033 & 0.009750 & 0.000094 & 0.000620 & 1769 & 160 & 132 & 9 \\ \hline
4 & 5 & 0.012315 & 0.003566 & 0.004621 & 0.058468 & 0.012228 & 0.062593 & 56 & 1481 & 888 & 70 \\ \hline
4 & 6 & 0.006298 & 0.001679 & 0.002073 & 0.028501 & 0.005290 & 0.038437 & 80 & 771 & 374 & 17 \\ \hline
4 & 10 & 0.002107 & 0.000246 & 0.000308 & 0.005958 & 0.000805 & 0.006085 & 175 & 183 & 134 & 4 \\ \hline
4 & 14 & 0.001721 & 0.000063 & 0.000083 & 0.003627 & 0.000224 & 0.001695 & 326 & 159 & 131 & 5 \\ \hline
4 & 18 & 0.001727 & 0.000029 & 0.000038 & 0.002954 & 0.000106 & 0.000801 & 590 & 156 & 128 & 6 \\ \hline
4 & 22 & 0.001890 & 0.000022 & 0.000030 & 0.003824 & 0.000083 & 0.000630 & 1067 & 155 & 127 & 6 \\ \hline
4 & 26 & 0.002192 & 0.000023 & 0.000030 & 0.006805 & 0.000084 & 0.000635 & 1542 & 158 & 131 & 6 \\ \hline
4 & 30 & 0.002460 & 0.000024 & 0.000032 & 0.010366 & 0.000086 & 0.000569 & 1784 & 162 & 134 & 9 \\ \hline
5 & 6 & 0.007367 & 0.001804 & 0.002318 & 0.036242 & 0.006253 & 0.044261 & 70 & 918 & 498 & 22 \\ \hline
5 & 10 & 0.002293 & 0.000272 & 0.000343 & 0.007225 & 0.000900 & 0.006806 & 166 & 206 & 146 & 4 \\ \hline
5 & 14 & 0.001729 & 0.000067 & 0.000088 & 0.003856 & 0.000245 & 0.001858 & 320 & 163 & 132 & 4 \\ \hline
5 & 18 & 0.001704 & 0.000031 & 0.000039 & 0.003162 & 0.000100 & 0.000762 & 597 & 155 & 128 & 5 \\ \hline
5 & 22 & 0.001835 & 0.000022 & 0.000029 & 0.003953 & 0.000083 & 0.000626 & 1107 & 150 & 123 & 5 \\ \hline
5 & 26 & 0.002133 & 0.000023 & 0.000030 & 0.006504 & 0.000082 & 0.000620 & 1541 & 156 & 128 & 6 \\ \hline
5 & 30 & 0.002341 & 0.000025 & 0.000033 & 0.009672 & 0.000091 & 0.000692 & 1810 & 157 & 129 & 7 \\ \hline
\end{tabular}
\end{table}

\begin{table}[htbp]
\centering
\small
\tabcolsep=2pt 
\caption{The relative errors of surface 2}
\begin{tabular}{|c|c|c|c|c|c|c|c|c|c|c|c|}
\hline
Deg & Cpts & Mean & Median & Trim. Mean & P95 & IQR & IQR Ratio & Neg & $>0.01$ & $>0.02$ & $>0.10$ \\ \hline
1 & 4 & 0.011681 & 0.007357 & 0.007923 & 0.038078 & 0.013299 & 0.096866 & 53 & 1992 & 861 & 8 \\ \hline
1 & 5 & 0.008941 & 0.005363 & 0.005723 & 0.029816 & 0.009651 & 0.082741 & 68 & 1450 & 505 & 2 \\ \hline
1 & 6 & 0.007361 & 0.004142 & 0.004469 & 0.023748 & 0.007610 & 0.060438 & 68 & 1074 & 345 & 5 \\ \hline
1 & 10 & 0.004061 & 0.001906 & 0.002053 & 0.012497 & 0.003589 & 0.028378 & 76 & 359 & 138 & 2 \\ \hline
1 & 14 & 0.003085 & 0.001160 & 0.001291 & 0.008493 & 0.002295 & 0.022131 & 80 & 216 & 130 & 1 \\ \hline
1 & 18 & 0.002651 & 0.000821 & 0.000927 & 0.007350 & 0.001660 & 0.016651 & 80 & 201 & 126 & 0 \\ \hline
1 & 22 & 0.002331 & 0.000629 & 0.000706 & 0.007036 & 0.001243 & 0.010647 & 77 & 194 & 125 & 4 \\ \hline
1 & 26 & 0.002247 & 0.000556 & 0.000629 & 0.006892 & 0.001093 & 0.008578 & 85 & 198 & 122 & 3 \\ \hline
1 & 30 & 0.002257 & 0.000419 & 0.000479 & 0.008295 & 0.000818 & 0.006353 & 87 & 211 & 136 & 5 \\ \hline
2 & 4 & 0.009576 & 0.005440 & 0.005980 & 0.031448 & 0.010843 & 0.080966 & 68 & 1567 & 588 & 9 \\ \hline
2 & 5 & 0.007158 & 0.003488 & 0.003915 & 0.025507 & 0.007597 & 0.059855 & 81 & 1041 & 364 & 5 \\ \hline
2 & 6 & 0.005633 & 0.002399 & 0.002723 & 0.019621 & 0.005512 & 0.044451 & 78 & 726 & 238 & 4 \\ \hline
2 & 10 & 0.002550 & 0.000680 & 0.000783 & 0.008868 & 0.001628 & 0.016315 & 90 & 217 & 123 & 0 \\ \hline
2 & 14 & 0.001829 & 0.000312 & 0.000350 & 0.004817 & 0.000679 & 0.006859 & 95 & 174 & 116 & 0 \\ \hline
2 & 18 & 0.001569 & 0.000187 & 0.000208 & 0.003494 & 0.000370 & 0.003732 & 101 & 161 & 111 & 0 \\ \hline
2 & 22 & 0.001543 & 0.000129 & 0.000148 & 0.004592 & 0.000270 & 0.002607 & 105 & 172 & 117 & 1 \\ \hline
2 & 26 & 0.001491 & 0.000103 & 0.000115 & 0.004230 & 0.000213 & 0.002149 & 125 & 171 & 116 & 0 \\ \hline
2 & 30 & 0.001608 & 0.000085 & 0.000100 & 0.004439 & 0.000196 & 0.001549 & 173 & 177 & 120 & 3 \\ \hline
3 & 4 & 0.009800 & 0.005578 & 0.006170 & 0.032289 & 0.010945 & 0.082763 & 89 & 1621 & 613 & 9 \\ \hline
3 & 5 & 0.007326 & 0.003553 & 0.003987 & 0.026884 & 0.007868 & 0.059385 & 84 & 1081 & 380 & 8 \\ \hline
3 & 6 & 0.005815 & 0.002445 & 0.002819 & 0.020594 & 0.005959 & 0.047621 & 91 & 778 & 268 & 3 \\ \hline
3 & 10 & 0.002715 & 0.000673 & 0.000788 & 0.010241 & 0.001711 & 0.016910 & 101 & 256 & 127 & 1 \\ \hline
3 & 14 & 0.001821 & 0.000264 & 0.000308 & 0.004978 & 0.000652 & 0.006036 & 113 & 166 & 117 & 2 \\ \hline
3 & 18 & 0.001513 & 0.000146 & 0.000168 & 0.003351 & 0.000333 & 0.003362 & 139 & 156 & 108 & 0 \\ \hline
3 & 22 & 0.001504 & 0.000099 & 0.000113 & 0.003816 & 0.000214 & 0.001831 & 147 & 167 & 116 & 1 \\ \hline
3 & 26 & 0.001527 & 0.000080 & 0.000092 & 0.004207 & 0.000177 & 0.001498 & 218 & 172 & 113 & 2 \\ \hline
3 & 30 & 0.001580 & 0.000067 & 0.000077 & 0.004673 & 0.000155 & 0.001489 & 328 & 174 & 118 & 1 \\ \hline
4 & 5 & 0.007512 & 0.003655 & 0.004139 & 0.027052 & 0.008093 & 0.063900 & 95 & 1135 & 391 & 8 \\ \hline
4 & 6 & 0.005893 & 0.002495 & 0.002895 & 0.021042 & 0.006094 & 0.048673 & 87 & 815 & 275 & 4 \\ \hline
4 & 10 & 0.002964 & 0.000784 & 0.000903 & 0.011131 & 0.001919 & 0.018907 & 113 & 277 & 136 & 1 \\ \hline
4 & 14 & 0.001779 & 0.000290 & 0.000343 & 0.005238 & 0.000750 & 0.007577 & 119 & 162 & 105 & 0 \\ \hline
4 & 18 & 0.001452 & 0.000147 & 0.000174 & 0.003567 & 0.000370 & 0.003736 & 129 & 151 & 105 & 0 \\ \hline
4 & 22 & 0.001490 & 0.000095 & 0.000109 & 0.003458 & 0.000221 & 0.001868 & 164 & 162 & 113 & 2 \\ \hline
4 & 26 & 0.001470 & 0.000071 & 0.000083 & 0.003823 & 0.000165 & 0.001583 & 240 & 162 & 115 & 1 \\ \hline
4 & 30 & 0.001533 & 0.000063 & 0.000075 & 0.005162 & 0.000159 & 0.001346 & 407 & 167 & 111 & 4 \\ \hline
5 & 6 & 0.006050 & 0.002642 & 0.003038 & 0.021277 & 0.006254 & 0.049771 & 105 & 834 & 275 & 3 \\ \hline
5 & 10 & 0.003199 & 0.000854 & 0.000983 & 0.011949 & 0.002208 & 0.018308 & 112 & 311 & 150 & 1 \\ \hline
5 & 14 & 0.001971 & 0.000323 & 0.000386 & 0.006692 & 0.000866 & 0.008437 & 127 & 175 & 110 & 1 \\ \hline
5 & 18 & 0.001500 & 0.000159 & 0.000188 & 0.003634 & 0.000409 & 0.004116 & 139 & 149 & 104 & 0 \\ \hline
5 & 22 & 0.001380 & 0.000094 & 0.000110 & 0.003016 & 0.000217 & 0.002126 & 168 & 153 & 104 & 1 \\ \hline
5 & 26 & 0.001442 & 0.000073 & 0.000085 & 0.003605 & 0.000168 & 0.001362 & 267 & 160 & 108 & 1 \\ \hline
5 & 30 & 0.001463 & 0.000060 & 0.000072 & 0.004616 & 0.000152 & 0.001537 & 399 & 168 & 112 & 0 \\ \hline
\end{tabular}
\end{table}

\end{document}